\documentclass[12pt]{article}

\usepackage{graphicx, amsmath, amssymb, amsthm, multicol, multirow, float, subcaption, verbatim, enumitem, rotating}
\usepackage[utf8x]{inputenc}
\usepackage{epstopdf}

\usepackage[boxed]{algorithm2e}

\usepackage{thmtools}
\usepackage{thm-restate}

\usepackage{cite}

\usepackage{color}   
\usepackage{hyperref}
\hypersetup{
	hidelinks,
    linktoc=all     
}

\newtheorem{theorem}{Theorem}[section]
\newtheorem{definition}[theorem]{Definition}

\theoremstyle{definition}
\newtheorem{remark}[theorem]{Remark}

\newtheorem{assumption}[theorem]{Assumption}

\numberwithin{equation}{section}

\newcommand{\trans}[1]{#1^{\mathsf{T}}}

\newcommand{\N}{\mathbb{N}}

\usepackage{bbm}

\def\cD{\mathcal{D}}

\def\cF{\mathcal{F}}

\def\cH{\mathcal{H}}

\def\cP{\mathcal{P}}
\def\cQ{\mathcal{Q}}

\def\cT{\mathcal{T}}

\def\bD{\mathbb{D}}
\def\bE{\mathbb{E}}
\def\bF{\mathbb{F}}

\def\bN{\mathbb{N}}

\def\bP{\mathbb{P}}
\def\bQ{\mathbb{Q}}
\def\bR{\mathbb{R}}

\def\bf1{\mathbf{1}}

\newcommand{\1}{\mathbbm{1}}            
\newcommand{\set}[1]{\{#1\}}            
\newcommand{\Set}[1]{\left\{#1\right\}} 

\DeclareMathOperator{\dif}{d \!}        

\DeclareMathOperator{\var}{\mathrm{V}@\mathrm{R}}           
\DeclareMathOperator{\tvar}{\mathrm{TV}@\mathrm{R}}         
\DeclareMathOperator{\glr}{\mathrm{GLR}}                    
\DeclareMathOperator{\raroc}{\mathrm{RAROC}}                \DeclareMathOperator{\ait}{\mathrm{AIT}}                     
\DeclareMathOperator{\evar}{\mathrm{EV}@\mathrm{R}}

\begin{document}
\title{Acceptability Maximization}
\author{Gabriela Kov\'a\v{c}ov\'a\thanks{Vienna University of Economics and Business, Institute for Statistics and Mathematics, Vienna A-1020, AUT, \url{gabriela.kovacova@wu.ac.at} and \url{birgit.rudloff@wu.ac.at}.} \and Birgit Rudloff\footnotemark[1] 
\and Igor Cialenco\thanks{Department of Applied Mathematics, Illinois Institute of Technology,  10 W 32nd Str, Building RE, Room 220, Chicago, IL 60616, USA, 
	\url{cialenco@iit.edu}}
}
\maketitle

	\vspace{-2em}

	
	{\footnotesize
		\begin{tabular}{l@{} p{350pt}}
			\hline \\[-.2em]
			\textsc{Abstract}: \ & The aim of this paper is to study the optimal investment problem by 	using coherent acceptability indices (CAIs) as a tool to measure the portfolio performance. We call this problem the acceptability maximization. First, we study the one-period (static) case, and propose a numerical algorithm that approximates the original problem by a sequence of risk minimization problems. The results are applied to  several important CAIs, such as the gain-to-loss ratio, the risk-adjusted return on capital and the tail-value-at-risk based CAI. In the second part of the paper  we investigate the acceptability maximization in a discrete time dynamic setup. Using robust  representations of CAIs in terms of a family of dynamic coherent risk measures (DCRMs), we establish an intriguing dichotomy: if the corresponding family of DCRMs is recursive (i.e. strongly time consistent) and assuming some recursive structure of the market model, then the acceptability maximization problem reduces to just a one period problem and the maximal acceptability is constant across all states and times. On the other hand, if the family of DCRMs is not recursive, which is often the case, then the acceptability maximization problem ordinarily is a time-inconsistent stochastic control problem, similar to the classical mean-variance criteria. To overcome this form of time-inconsistency, we adapt to our setup the set-valued Bellman's principle recently proposed in~\cite{KovacovaRudloff2019} applied to two particular dynamic CAIs - the dynamic risk-adjusted return on capital and the dynamic gain-to-loss ratio. The obtained theoretical results are illustrated via numerical examples that include, in particular, the computation of the intermediate mean-risk efficient frontiers.  
			
			\\[0.5em]
			\textsc{Keywords:} \ & acceptability index, acceptability maximization, optimal portfolio, gain-to-loss ratio,  dynamic performance measures, tail-value-at-risk, set-valued Bellman principle 
			 \\
			\textsc{MSC2010:} \ &  91G10,  93E20, 93E35, 90C35,  91B30, 49L20 \\[1em]
			\hline
		\end{tabular}
	}

\section{Introduction}
The renowned Sharpe ratio, introduced in \cite{Sharpe1964}, besides being one of the best known tools in measuring the performance of financial portfolios, also played an important role in developing the  portfolio optimization theory. It is well-known that one of the major shortcomings of the Sharpe ratio, as a performance measure, is its lack of monotonicity, i.e. a portfolio with strictly larger future gains may have a smaller Sharpe ratio. Over the years, in particular to overcome this drawback, a number of performance measures were introduced such as the Gini ratio~\cite{ShalitYityaki84}, the MAD ratio~\cite{KonnoYamazaki91}, the minimax ratio~\cite{Young98}, the gain-to-loss ratio~\cite{BernardoLedoit00}, the Sortino-Satchell ratio~\cite{SortinoSatchell01}, among many others; see \cite{CheriditoKromer2012,OrtobelliETAL05} for a comprehensive survey of ratio type performance measures. This naturally raises the question which properties -- and, consequently, which performance measures -- are desirable. Cherny and Madan~\cite{ChernyMadan2009} took an axiomatic approach to performance measurement, in line with the classical axiomatic approach to risk measurement by Artzner et al.~~\cite{ArtznerDelbaenEberHeath1999}. In~\cite{ChernyMadan2009} the authors introduce the concept of a coherent acceptability index (CAI) as a measure of performance that satisfies a set of desirable properties -- monotonicity, quasi-concavity, scale invariance and the Fatou property. It was proved that an unbounded CAI $\alpha$ admits a robust representation  in terms of a family of coherent risk measures (CRMs) $(\rho^x)_{x\in(0,\infty)}$ given by 
\begin{align}
\label{AIRM_intro}
\alpha(D) = \sup \{ x \in (0, \infty) : \rho^x (D) \leq 0 \}.
\end{align}
Equivalently, the representation~\eqref{AIRM_intro} can be stated in terms of a family of acceptance cones, or in terms of a family of sets of probability measures. The concept of a coherent acceptability index in a dynamic setup was first studied in~\cite{BCZ2010}, and consequently in~\cite{BCDK2013,BCC2014,RosazzaGianinSgarra2012,BiaginiBion-Nadal2012}. 

The main goal of this paper is to study the optimization problem  of the form 
\begin{equation}\label{maxAI_intro}
\max\limits_{D \in \mathcal{D}} \; \alpha(D),
\end{equation}
where $\alpha$ is a static or dynamic CAI and $\cD$ is a set of feasible positions. We refer to this problem as \textit{acceptability maximization}. 
The obtained results contribute to the rich literature on optimal portfolio selection or optimization of performance, such as classical mean-variance portfolio analysis or Sharpe ratio maximization (cf.~\cite{AgarwalNaik2004,GoetzmannETAL02}). Clearly, acceptability maximization is a practically important and natural problem to study. Closest to the spirit of our study is \cite{EberleinMadan14}, where the authors solve an acceptability maximization problem for some specific choices of static CAIs (AIMIN, AIMAX, AIMINMAX and AIMAXMIN) which are represented through families of distortion functions (or Choquet integrals). To the best of our knowledge, this is the only work that considers optimal control problems with CAI criteria. The proposed solution in \cite{EberleinMadan14} fundamentally relies on the representation of CAIs in terms of distortion functions. 

In contrast to \cite{EberleinMadan14}, in this work we mostly exploit the representation~\eqref{AIRM_intro}, and we consider both the static and the dynamic case.  We start by considering the one-period (static) setup, presented in Section~\ref{sec_static}. First, we recall some key definitions and relevant results on CAIs (Section~\ref{sec_staticAI}), and then we propose a numerical algorithm for solving \eqref{maxAI_intro} that approximates the maximal acceptability by a sequence of risk minimization problems (Section~\ref{sec_AImaxStatic}). In Section~\ref{sec_examples} we apply this algorithm to several important CAIs, such as the gain-to-loss ratio, the risk-adjusted return on capital and the tail-value-at-risk based CAI.    

Undoubtedly, for many practical purposes, acceptability or performance needs to be measured in a multi-period setting where the investor dynamically rebalances her portfolio. We study this in Section~\ref{sec_dynamic}. As was noted in \cite{BCC2014} and later systemically addressed in \cite{BCP2014,BCP2014a}, the time consistency property stays at the heart of the matter when studying dynamic coherent acceptability indices (DCAIs) and their robust representations of the form \eqref{AIRM_intro} in terms of families of dynamic coherent risk measures (DCRMs) $(\rho_t^x)^{x>0}_{t=0,\ldots,T}$.   
We prove that if for every $x>0$ the DCRM $(\rho_t^x)_{t=0,\ldots,T}$ is recursive (or strongly time consistent as a DCRM), and assuming some recursive structure of the underlying market model, then the maximal acceptability is constant across all times and states. Thus, in this case, it is enough to solve the corresponding optimization problem only once, at one state and one period of time; see Section~\ref{sec_recursiveRM} for details. 
However, in many relevant examples of DCAIs,  $(\rho_t^x)_{t=0,\ldots,T}$ are not  strongly  time consistent, and, similar to the classical control problems with mean-variance criteria, problem \eqref{maxAI_intro} is time-inconsistent in the sense that the (naive) dynamic programming principles does not hold true. To overcome this challenge, we adapt to our setup the set-valued Bellman's principle recently proposed in~\cite{KovacovaRudloff2019}. This approach also provides the intermediate mean-risk efficient frontiers, in the spirit of the mean-variance efficient frontier. Within this approach, we consider two specific performance measures - the dynamic risk-adjusted return on capital (Section~\ref{sec:dRAROC}) and the dynamic gain-to-loss ratio (Section~\ref{sec:dGLR}). 
The majority of the proofs are deferred to the appendices.  

As mentioned above, this is the first attempt to study stochastic control problems with dynamic CAI criteria. While the static case is now relatively well understood, the acceptability maximization problem in a dynamic setup appears to be an interesting research area, with many open problems, primarily due to the time-inconsistent nature of such problems, as argued in this manuscript. In particular, it would be far-reaching to develop a Bellman's principle of optimality for a class of DCAIs, beyond particular examples. In addition, from a practical point of view, it would be important to study the acceptability maximization problem for a larger classes of indices, for example not necessarily coherent ones. The authors plan to treat these problems in future works.

\section{Acceptability Maximization in the Static Setting}
\label{sec_static}
In this section, we consider the static setting, before moving to the dynamic one in Section~\ref{sec_dynamic}. 
First, we recall the definition of a coherent acceptability index and its connection to coherent risk measures. This serves as our framework for studying the maximization of performance, i.e., acceptability maximization. In Subsection~\ref{sec_AImaxStatic} we provide a way to solve the acceptability maximization problem through a sequence of risk minimizations. At the end of the section, we provide examples of this approach.  Proofs can be found in Appendix~\ref{AppendixA}.

\subsection{Static Coherent Acceptability Indices and Risk Measures}
\label{sec_staticAI}
We start by briefly reviewing the notion of a (static) coherent acceptability index (CAI) and its connection to (static) coherent risk measures (CRMs), following~\cite{ChernyMadan2009}. 
The concept of acceptability was developed as a methodology to define axiomatically minimal desirable properties of a functional that is meant to measure or assess the performance of a financial position or trading portfolio. As usual, we consider an underlying probability space $(\Omega, \mathcal{F}, \mathbb{P})$, and we denote by $L^\infty:=  L^\infty (\Omega, \mathcal{F}, \mathbb{P})$ the space of essentially bounded random variables on this space. In what follows, all equalities and inequalities between random variables will be understood in a $\bP$ almost surely sense. In this section, an element $D\in L^\infty$ can be viewed as a discounted terminal cash flow of a zero-cost self-financed portfolio, or the terminal profit and loss (P\&L) of a financial position.  The mapping $D \mapsto \alpha(D) \in [0, \infty]$ assigns to the portfolio $D$ the degree of its acceptability, with higher values corresponding to more desirable positions.

\begin{definition}
A \textbf{coherent acceptability index (CAI)} is a function $\alpha: L^\infty \rightarrow [0, \infty]$ satisfying for all positions $D, D' \in L^\infty$ and any level $x \in (0, \infty)$
\begin{enumerate}
\item[(A1)]  Monotonicity: if $D \leq D'$, then $\alpha(D) \leq \alpha(D')$,
\item[(A2)] Scale invariance: $\alpha(\lambda D) = \alpha(D)$ for all $\lambda>0$,
\item[(A3)] Quasi-concavity: if $\alpha(D) \geq x$ and $\alpha(D') \geq x$, then $\alpha(\lambda D + (1-\lambda)D') \geq x$ for all $\lambda \in [0,1]$,
\item[(A4)] Fatou property: if $\vert D_n \vert \leq 1, \alpha(D_n) \geq x$ for all $n\geq 1$, and $D_n \overset{\bP}{\rightarrow} D$, then $\alpha(D)\geq x$.
\end{enumerate}
\end{definition}

In~\cite{ChernyMadan2009} four further properties -- law invariance, consistency with the second order stochastic dominance, arbitrage consistency and expectation consistency -- are discussed. These are not required for the coherent acceptability index, but, as the authors argue, they are desirable.  
Additionaly, coherent acceptability indices are closely related to coherent risk measures, a concept introduced in~\cite{ArtznerDelbaenEberHeath1999}. 
\begin{definition}
A \textbf{coherent risk measure (CRM)} is a function $\rho: L^\infty \rightarrow \mathbb{R}$  satisfying for any $D, D' \in L^\infty$
\begin{enumerate}
\item[(R1)] Monotonicity: if $D \leq D'$, then $\rho(D) \geq \rho(D')$,
\item[(R2)] Positive homogeneity: $\rho(\lambda D) = \lambda \rho(D)$ for all $\lambda>0$,
\item[(R3)] Translation invariance: $\rho(D + k) = \rho(D) - k$ for all $k \in \mathbb{R}$,
\item[(R4)] Subadditivity: $\rho(D + D') \leq \rho(D) + \rho(D')$, 
\item[(R5)] Fatou property: if $\vert D_n \vert \leq 1$, for all $n\geq 1$, and $D_n \overset{\bP}{\rightarrow} D$, then $\rho(D)\leq \liminf_{n\to\infty} \rho(D_n)$.
\end{enumerate}
\end{definition}
A family of coherent risk measures $(\rho^x)_{x \in (0, \infty)}$ is called \textbf{increasing} if $x \geq y >0$ implies $\rho^x (D) \geq \rho^y (D)$ for any $D \in L^\infty$. 

As was proved in~\cite{ChernyMadan2009}, there is a strong connection between CAIs and increasing families of CRMs. Namely, the following robust representation type result holds true: A map $\alpha: L^\infty \rightarrow [0, \infty]$, unbounded from above, is a coherent acceptability index if and only if there exists an increasing family of coherent risk measures $(\rho^x)_{x \in (0, \infty)}$ such that
\begin{align}\label{AIRM}
\alpha(D) = \sup \{ x \in (0, \infty) : \rho^x (D) \leq 0 \}
\end{align}
with the convention $\sup \emptyset = 0$. Equivalently, the representation \eqref{AIRM} can be formulated in terms of a family of acceptance sets, or an increasing family of sets of probability measures associated with dual representations of CRMs (see also Section~\ref{sec_examples}). For various degrees of generalization of \eqref{AIRM} see, for instance, 
\cite{MadanSchoutens2016,BCDK2013,BCC2014,RosazzaGianinSgarra2012,BiaginiBion-Nadal2012}.

\subsection{Algorithm for Acceptability Maximization: the Static Case}\label{sec_AImaxStatic}
We fix a set  $\cD\subset L^\infty$ of available or feasible positions. For example, $\cD$ could be the P\&Ls of portfolios that satisfy certain trading or other constraints. Our aim is to identify among the available positions the ones with the highest degree of acceptability. Namely, we wish to solve the following optimization problem,
\begin{align}\tag{$A$}
\label{maxAI}
\max\limits_{D \in \mathcal{D}} \; \alpha(D).
\end{align}
We denote the maximal acceptability by
\begin{align}\label{alpha*}
\alpha^* := \sup\limits_{D \in \mathcal{D}} \; \alpha(D)
\end{align}
and the set of maximally acceptable (optimal) portfolios by
$\mathcal{D}^* := \{ D \in \mathcal{D} : \alpha(D) = \alpha^* \}.$
Generally speaking, a maximally acceptable portfolio may not exist, i.e.~the set $\mathcal{D}^*$ is empty if $\alpha^*$ is not attained as a maximum. For $\epsilon > 0$ we define the set of $\epsilon$-optimal positions
\begin{align*}
\mathcal{D}^\epsilon := \{ D \in \mathcal{D} : \alpha(D) \geq \alpha^* - \epsilon \}.
\end{align*}
The following Lemma summarizes the properties of these sets.
\begin{restatable}{lemma}{lemmaD} \label{lemma_D} 
\mbox{}	
\begin{enumerate}
\item If the feasible set $\mathcal{D}$ is convex, then the sets $\mathcal{D}^*$ and $\mathcal{D}^\epsilon$ are convex, for any  $\epsilon > 0$.
\item The sets $\mathcal{D}^*$ and $\mathcal{D}^\epsilon$ are nested: $\mathcal{D}^* \subseteq \mathcal{D}^{\epsilon_1} \subseteq \mathcal{D}^{\epsilon_2}$, for any $\epsilon_2 > \epsilon_1 > 0$.
\item 
$ \mathcal{D}^* = \bigcap\limits_{\epsilon > 0} \mathcal{D}^\epsilon$.
\end{enumerate}
\end{restatable}
\begin{proof}
The proof is deferred to Appendix~\ref{AppendixA}.
\end{proof}

To solve the maximization problem \eqref{maxAI}, we will use the robust representation \eqref{AIRM}, with $(\rho^x)_{x \in (0, \infty)}$ denoting the corresponding increasing family of CRMs. For a given level $x > 0$ we consider the problem of minimizing risk over the feasible set $\mathcal{D}$, 
\begin{align}\tag{$P_x$}\label{minRM}
\min\limits_{D \in \mathcal{D}} \rho^x(D).  
\end{align}
We denote the optimal value of the risk minimization problem~\eqref{minRM} by $p(x) := \inf\limits_{D \in \mathcal{D}} \rho^x(D)$,  and its optimal solution by $D^x \in \arg\min\limits_{D \in \mathcal{D}} \rho^x(D),$ assuming that the infimum $p(x)$ is attained. In what follows we make the following standing assumption:

\begin{assumption}\label{assumption}
	The acceptability index $\alpha$ is unbounded from above, and for every $x \in (0,\infty)$ the risk minimization problem~\eqref{minRM} attains its minimum.
\end{assumption}
The unboundness from above of $\alpha$ is usually satisfied  in all practically important cases, while attainability of the min in \eqref{minRM} can be guaranteed, for example, by assuming that $\cD$ is compact. With this at hand, and in view of \eqref{AIRM}, we note that: 
\begin{itemize}
\item if the risk minimization problem~\eqref{minRM} has a positive optimal value, then no portfolio in $\mathcal{D}$ has a degree of acceptability above $x$;
\item if the risk minimization problem~\eqref{minRM} has a non-positive
optimal value, then some portfolio in $\mathcal{D}$ has a degree of acceptability of at least $x$.  
\end{itemize}

The next result summarizes the above observations, which are used in developing Algorithm~\ref{alg} that solves the acceptability maximization problem \eqref{maxAI}.

\begin{restatable}{lemma}{lemmaax}
\label{lemma_ax}
Let Assumption~\ref{assumption} hold and let $x \in (0, \infty)$.
\begin{enumerate}
\item If $p(x) > 0$, then $x \geq \alpha^*.$
\item If $p(x) \leq 0$, then $x \leq \alpha^*.$
\item If $x < \alpha^* $, then $p(y) \leq 0$ for all $y \leq x.$
\item If $x > \alpha^*$, then $p(y) > 0$ for all $y \geq x.$
\end{enumerate}
\end{restatable}
\begin{proof}
The proof is given in Appendix~\ref{AppendixA}.
\end{proof}

The main idea of the proposed numerical solution of \eqref{maxAI} is to approximate the maximal acceptability by a sequence of risk minimization problems \eqref{minRM} for some appropriately chosen levels - a variation of the bisection method that will find a pair $(x,D)$  maximizing the level $x$ while satisfying $\rho^x(D)\leq0$. 
First, find two levels with opposite signs of the minimal risk $p(x)$, namely find a lower and upper bound on the maximal acceptability $\alpha^*$. Then, iteratively decrease the distance between the two bounds. This replaces one acceptability maximization problem~\eqref{maxAI} with a sequence of risk minimization problems~\eqref{minRM}. This becomes particularly useful if the acceptability maximization is complicated and the risk minimization is easier to solve. Note that if the feasible set $\mathcal{D}$ is convex, then the risk minimization becomes a convex optimization problem.

\begin{algorithm}[h] \small
\caption{Approximating maximal acceptability $\alpha^*$ via risk minimization}
\SetAlgoLined
\SetKw{AND}{\textbf{and}}
\SetKw{OR}{\textbf{or}}
 \KwIn{initial level $x_0 \in (0, \infty)$, max. iterations $\bar{M} \in \mathbb{N}$ of Step 1, tolerance $\epsilon > 0$.}
 \textit{Step 1: Find an initial interval $x_L \leq \alpha^* \leq x_U$} \\
 Set $n := 0; \, x_L := 0; \, x_U := \infty; \, \bar{D} :=\textrm{null}$ \;
 \While{($x_L == 0$ \OR $x_U == \infty$) \AND ($n < \bar{M}$)}{
  \eIf{ The optimal value of (P$_{x_n}$) is positive }{
   $x_U := x_n,$ select $x_{n+1} := x_U / 2$ \;
   }
   {
   $x_L := x_n,$ select $x_{n+1}:= 2 \cdot x_L,$ assign $\bar{D} := D^{x_n}$ \;
  }
  $n := n+1$ \;
 }
 
 \textit{Step 2: Decrease the length of the interval $[x_L, x_U]$ (via bisection)} \\
 \While{($x_U - x_L \geq \epsilon$) \AND ($n < \bar{M}$)}{
  Select $x := (x_U + x_L)/2$ \;
  \eIf{ The optimal value of \eqref{minRM} is positive }{
   $ x_U := x$ \;
   }
   {
   $ x_L := x,$ assign $\bar{D} := D^x$ \;
  }
 }
 \KwOut{Interval $[x_L, x_U]$ as an approximation to $\alpha^*$, portfolio $\bar{D}$ as approximately optimal one}
 \label{alg}
\end{algorithm}

The next result summarizes the key features of Algorithm~\ref{alg}.  
\begin{restatable}{lemma}{lemmaalg}
\label{lemma_alg}
Suppose that Assumption~\ref{assumption} holds, and let $x_0 \in (0, \infty)$ be the initial (seed) value, $\bar{M} \in \mathbb{N}$ be the maximal number of iterations (of Step 1) and $\epsilon > 0$ be the tolerance level. Denote
$$
\underline{\alpha} := x_0 \cdot 2^{-\bar{M}+1}, \quad \overline{\alpha} := x_0 \cdot 2^{\bar{M}-1}.
$$
\begin{enumerate}
\item If $\alpha^* \in [0, \underline{\alpha})$, then  Algorithm~\ref{alg} returns $\underline{\alpha}$ as an upper bound for $\alpha^*$, no acceptable portfolio is found.
\item If $\alpha^* \in (\overline{\alpha}, \infty]$, then Algorithm~\ref{alg} returns $\overline{\alpha}$ as a lower bound for $\alpha^*$ and a portfolio with (at least) this degree of acceptability.
\item If $\alpha^* \in (\underline{\alpha}, \overline{\alpha})$, then
\begin{enumerate}
\item Algorithm~\ref{alg} returns bounds $x_L$ and  $x_U$ such that $x_L \leq \alpha^* \leq x_U$ and $x_U - x_L < \epsilon$,
\item Algorithm~\ref{alg} returns an $\epsilon$-solution, i.e.~$\bar{D} \in \mathcal{D}^\epsilon$,
\item Step 2 of Algorithm~\ref{alg} terminates after at most $\left \lceil{\log_2 \frac{x_0}{\epsilon} + \bar{M} - 2}\right \rceil $ iterations.
\end{enumerate}
\end{enumerate}
\end{restatable}
\begin{proof}
The proof is postponed to Appendix~\ref{AppendixA}.
\end{proof}

\begin{remark}\label{remark_modified0}
	
Several comments are in order: 
\begin{enumerate}
	\item[(i)] 
In Step 1, instead of the halving ($x_{n+1} := x_U / 2$), respectively the doubling ($x_{n+1}:= 2 \cdot x_L$), one could select any $x_{n+1} < x_U,$ respectively $x_{n+1} > x_L.$ Similarly, in Step 2 one could replace the bisection with any choice of $x \in (x_L, x_U).$ The results of Lemma~\ref{lemma_alg} would differ in the interval $(\underline{\alpha}, \overline{\alpha})$ on which $\alpha^*$ is identified and the number of iterations.

\item[(ii)]  The case $\alpha^* = x$ is not included in Lemma~\ref{lemma_ax}. If $\alpha^* = x$, then we can only say that $p(y) \leq 0$ for all $y < x$. The sign of $p(\alpha^*)$ is not clear, since we do not know if the suprema in~\eqref{alpha*} and~\eqref{AIRM} are attained as maxima. Consequently, in the cases $\alpha^* \in \{ \underline{\alpha}, \overline{\alpha} \}$ we cannot derive the behaviour of Algorithm~\ref{alg}.

\item [(iii)]
Assumption~\ref{assumption} allows us to merge the case $p(x) = 0$ with the case $p(x) < 0$. Without it, we would need to distinguish between attained and not attained infimum for $p(x) = 0$. If the infimum $p(x) = 0$ is attained for some portfolio $\tilde{D} \in \mathcal{D}$, then $\rho^x (\tilde{D}) = 0$ and $x$ is a lower bound on $\alpha^*$. If the infimum is not attained, then for all positions $D \in \mathcal{D}$ we have $\rho^x(D) > 0,$ so $x$ is an upper bound on $\alpha^*.$ This distinction would need to be built into Algorithm~\ref{alg}. Alternatively, if $p(x) = 0$, this level $x$ could be discarded and the iteration repeated with some other choice of level in the appropriate interval. However, it is not clear how many such repetitions might be needed.

\item[(iv)] 
Instead of specifying the maximal number of iterations $\bar{M}$ and the initial value $x_0$ we could directly specify the interval $[\underline{\alpha}, \overline{\alpha}]$ in which the optimal $\alpha^*$ would be searched. Then, Step 1 of the algorithm would need to check the signs of $p(\underline{\alpha})$ and $p(\overline{\alpha})$ and terminate immediately if $\alpha^*$ lies in the interval $[0, \underline{\alpha})$ or $(\overline{\alpha}, \infty].$ Step 2 would require $\left \lceil{ \log_2 \frac{\overline{\alpha} - \underline{\alpha}}{\epsilon} }\right \rceil $ iterations to terminate, assuming bisection steps. This might be of interest especially if the risk measure is well defined also for the limiting cases $x = 0$ and $x = \infty,$ see Section~\ref{sec_examples} for an example.

\end{enumerate}

\end{remark}

Algorithm~\ref{alg} with tolerance $\epsilon$ outputs an $\epsilon$-solution to the acceptability maximization problem, an element of the set $\mathcal{D}^\epsilon.$ We also know that the sets of $\epsilon$-solutions are nested and intersect in the set of optimal solutions. Therefore, it is natural to ask about the convergence of the algorithm output as the tolerance $\epsilon$ vanishes. As the next result shows, such convergence holds true if the feasible set $\mathcal{D}$ is compact. Generally speaking, for a non-compact $\mathcal{D}$ it is possible to construct counter-examples, where the $\epsilon$-solutions $D^\epsilon \in \mathcal{D}^\epsilon$ converge to an (infeasible) element outside of a (non-empty) optimal set $\mathcal{D}^*,$ or diverge.

\begin{restatable}{lemma}{lemmaconvergence}\label{lemma_convergence}
Let $\alpha^* \in (\underline{\alpha}, \overline{\alpha})$ and suppose that the feasible set $\mathcal{D}$ is a compact set w.r.t. the topology of convergence in probability. Let $\{ D^{\epsilon_n} \}_{n \in \N}$ be a sequence of solutions outputed from Algorithm~\ref{alg} for a sequence of tolerances $\{ \epsilon_n \}_{n \in \N}$ with $\lim\limits_{n \rightarrow \infty} \epsilon_n = 0$. Then $\{ D^{\epsilon_n} \}_{n \in \N}$ has an $\bP$-a.s. convergent  subsequence whose limit belongs to the set of optimal solutions $\mathcal{D}^*.$
\end{restatable}
\begin{proof}
The proof is given in Appendix~\ref{AppendixA}.
\end{proof}

\subsection{Numerical Examples}\label{sec_examples}
We will illustrate the proposed algorithm with three examples of CAIs: the acceptability index corresponding to the tail-value-at-risk ($\ait$), the gain-to-loss ratio ($\glr$) and the risk-adjusted return on capital ($\raroc$). First, we define these CAIs as well as identify the families of risk measures from the robust representation \eqref{AIRM}. 

\begin{enumerate}
\item
The $\tvar$ at level $q \in (0,1)$ is defined as
\begin{align*}
 \tvar_q (D) = \frac{1}{q} \int\limits_0^q \var_p (D) \dif p,
\end{align*}
where $\var_p(D) := \inf \{ r \in \mathbb{R} \; : \; \mathbb{P}(D+r <0) \leq p \}$ is the value-at-risk at level $p\in (0,1)$. It is well known that $\var$ is not a coherent risk measure, while $\tvar$ is a coherent risk measure. Moreover, using $\tvar$ as a family of CRMs, we define the the acceptability index
\begin{align*}
\ait(D) := \sup \left\lbrace x \in (0, \infty) \; : \; \tvar_{\frac{1}{1+x}} (D) \leq 0  \right\rbrace. 
\end{align*}
It is easy to show that $\ait$ indeed is a CAI that is also law invariant, consistent with second order stochastic dominance,  and arbitrage and expectation consistent; for more details see \cite[Section~3.5]{ChernyMadan2009}. 
However, one notable drawback of $\ait$ is that it ignores the gains and only takes into account the tail corresponding to losses.

\item
The gain-to-loss ratio is a CAI, popular among practitioners, and defined as the ratio of the mean and the expectation of the negative tail, namely 
\begin{align*}
\glr(D) := \frac{(\mathbb{E}[D])^+}{\mathbb{E}[D^-]},
\end{align*}
where $D^- = \max\set{0,-D}, D^+=\max\set{D,0}$ and the convention $\frac{a}{0} := +\infty$ for all $a\geq 0$ is used. For additional properties of $\glr$ see, for instance, \cite[Section 3.2]{ChernyMadan2009}.  
We recall several representations of GLR, useful for our purposes. The representation of the form \eqref{AIRM} is given in terms of expectiles. 
We recall that the expectile $e_q(D)$ of a random variable $D$ at level $q \in (0,1)$ is defined by a first-order condition
\begin{align*}
q \mathbb{E}[(D - e_q(D))_+] = (1-q) \mathbb{E}[(D - e_q(D))_-],
\end{align*}
or as a minimizer of an asymmetric quadratic loss, see~\cite{BelliniBernardino15} for more details. The expectile-$\var$, 
$$
\evar_q (D) := -e_q(D),
$$
for $q\leq1/2$, is an increasing family of CRMs. One can show, for example by using that   $\mathcal{A}_{\evar_q} = \left\lbrace D \, \vert \, \frac{\mathbb{E}(D^+)}{\mathbb{E}(D^-)} \geq \frac{1-q}{q}  \right\rbrace$  is the acceptance set of $\evar_q$, that the representation \eqref{AIRM} for $\glr$ has the form 
\begin{align*}
\glr (D) := \sup \left\lbrace x \in (0, \infty) \; : \; \evar_{\frac{1}{2+x}} (D) \leq 0  \right\rbrace.
\end{align*}
Alternatively, one can use the system of supporting kernels corresponding to $\glr$, as well as the explicit form of the extreme measures, cf. \cite[Proposition~2]{ChernyMadan2009}. 
It is also clear that for $x > 0$,  $\glr(D) \geq x$ if, and only if, $\mathbb{E}[-D] + x \mathbb{E}[D^-] \leq 0$, which can be conveniently used for computation purposes. This is not linked to the robust representation \eqref{AIRM} since the mappings $D \mapsto \mathbb{E}[-D] + x \mathbb{E}[D^-]$ are not CRMs (for instance they are not translation invariant). 
Finally, we remark that there is another popular version of GLR, defined as $\overline \glr (D) = \frac{\mathbb{E}[D^+]}{\mathbb{E}[D^-]}$. This version of $\glr$ is also monotone, scale invariant and has the Fatou property, but lacks quasi-concavity, and thus $\overline{\glr}$ is not coherent. The two are connected via $\glr (D) = \max \{\overline{\glr} (D) - 1, 0 \}$.

\item
The risk-adjusted return on capital, similar to $\glr$, is a reward-risk type ratio, formally defined as
\begin{align}
\raroc(D) := \dfrac{(\mathbb{E}[D])^+}{(\pi(D))^+},
\label{eq:RAROC}
\end{align}
where $\pi$ is a fixed CRM. The corresponding family of CRMs is given by 
\begin{align*}
 \rho^x (D) = \min \Set{ \pi (D),  \ \frac{1}{1+x} \mathbb{E}[-D] + \frac{x}{1+x} \pi (D) } .
\end{align*}
For risk measures $\pi$ satisfying $\mathbb{E} [-D] \leq \pi (D)$ this simplifies to $\rho^x (D) =  \frac{1}{1+x} \mathbb{E}[-D] + \frac{x}{1+x} \pi (D)$. 
For more details see \cite[Section 3.4]{ChernyMadan2009}. In our numerical examples we will use $\raroc$ with $\pi = \tvar_{0.01}$, also known as the stable tail-adjusted return ratio (see, for instance,~\cite{MartinETAL03}). 
\end{enumerate}

In our numerical examples below, we maximize the acceptability index over the set of profits and losses that are possible to attain by investing in the market with $d$ (risky) assets. Without loss of generality, thanks to scale invariance of CAIs, we fix the initial investment to $1$. For numerical tractability, we assume that $\Omega$ is finite. 
The (gross or total) returns $S^j_1/S_0^j, \ j=1,\ldots,d$, of these $d$ assets are modeled as a matrix $R\in\bR^{d\times|\Omega|}$. Then, the sets of available profits and losses, with or without short-selling, become 
$$
\mathcal{D} = \{ \trans{R} h -1 \; : \; \trans{\mathbf{1}} h = 1 \} \ \text{ or } \  \mathcal{D} = \{ \trans{R} h -1  \; : \; \trans{\mathbf{1}} h = 1, h \geq 0 \}, 
$$
where $\mathbf{1}= \trans{(1,1,\ldots,1)}$. So, $h$ corresponds to the trading strategy (the amount invested in each asset) and $D(h)=\trans{R} h -1$ to the corresponding terminal P\&L.

To illustrate the main features of the proposed algorithm, we first consider a toy market model,  consisting of $d=2$ assets and  with returns given in Table~\ref{tab1}, Panel A. Generally speaking,  it is reasonable to select the input parameters such that that $\epsilon \leq \underline{\alpha}$. 
Panel B of Table~\ref{tab1} summarizes the iterations of the algorithm with the following input parameters:  the starting (seed) acceptability level set to $x_0 = 2$, the tolerance $\epsilon = 10^{-4}$ and the maximal number of iterations $\bar{M} = 15$ of Step 1. The algorithm outputs bounds on the maximal acceptability and an $\epsilon$-optimal solution $h^\epsilon$; see the last two rows of Table~\ref{tab1}, Panel B.  For each of the three acceptability indices, the optimal portfolio puts more weight on the first asset, with $\ait$ being the most balanced and $\raroc$ being the most extreme. This is because the first asset carries (in some sense) less risk, although at the cost of lower mean return than the second one. All three considered CAIs are loss based measures, but each in a different way. The $\ait$ measures how far and how deep into the tail the losses can go. The optimal position for $\ait$ balances the return in the second and the third state of the world. The $\glr$ treats loss directly through the expectation of effective losses (the negative part of the P\&L). Thus, the corresponding optimal position is in the range where the portfolio return is negative in only one state of the world. Since we are using $\tvar_{0.01}$ in defining the $\raroc$, only the worst-case scenario (state of the world) is considered, which is the reason why the corresponding optimal position relies heavily on the first asset, for which the worst-case loss is lower.  We also remark that in this market model, the short-selling constraints do not change the results.

For the sake of completeness, we also show the iteration of the modified algorithm outlined in Remark~\ref{remark_modified0}(iv). We use the fact that for each of the three indices -- $\ait, \  \glr, \ \raroc$ -- the corresponding risk measures are well-defined for the limiting parameter values $x = 0$ and $x = \infty.$ Since the bisection cannot be done on an interval of infinite length, we index the families of risk measures by a parameter $q = \frac{1}{1+x}$ on a bounded interval $[0,1],$ or, respectively by $q = \frac{1}{2+x}$ on $[0, 0.5].$ Then, the bisection is performed with respect to the parameter $q$. The iterations  for $\glr$ are presented in Table~\ref{tab3}, see the modified algorithm. This modification avoids the risk of failing to find a lower or upper bound for a badly chosen starting point $x_0$ (compare to Table~\ref{tab2}). Moreover, zero and infinite acceptability are often determined after solving two risk minimization problems, instead of $\bar{M}.$ On the other hand, one needs to treat the tolerance parameter $\epsilon$ carefully: although the bisection is performed on the parameter $q$, the termination criterion needs to be set on $x$ in order for the error not to be distorted (see Table~\ref{tab3}). A mixed version of the algorithm is also provided -- it switches to a bisection on the original parameter $x$, as soon as a finite upper bound is found. The iterations for $\glr$ are also given in Table~\ref{tab3}, see the mixed algorithm. In addition, we also make the following slight modification to the algorithm: at each iteration a risk minimization problem $P(x)$ is solved, finding its optimal solution $D^x$. If the considered level is found to be a lower bound then the maximal acceptability, then the position $D^x$ is used for updating the optimal solution of the acceptability maximization. Otherwise, it is not used at all. One can easily see that given a fixed position $D^x$, all levels $y$ satisfying $\rho^y (D^x) \leq 0$ are lower bounds on the maximal acceptability. Hence, if it is relatively easy to find a level $y$ such that $\rho^y (D^x) = 0$, then it can be used to update the lower bound. For the CAIs used in this example  such level $y$ can be found without any further optimization. We refer to this modification of the original algorithm as zero-level version.

We also run the proposed algorithm on a more realistic market model, consisting of  $d = 10$ stocks and $\vert \Omega \vert = 1000$ states of the world. The return matrix is obtained as draws from a multivariate Student's $t$-distribution. In Table~\ref{tab2} we report the results for various input parameters $x_0, \epsilon$ and $\bar{M}$. The results are intuitively clear, and expected: the distance of the initial guess $x_0$ from the true $\alpha^*$ affects the number of iterations needed to find the upper and lower bound (Step 1). The tolerance $\epsilon$ determines the number of iterations in Step~2.  We also note that for a badly chosen starting point the algorithm can fail to find a lower or an upper bound unless $\bar{M}$ is increased. We also note that the maximal acceptability differs in a market with and without short-selling, but the impact of the parameters is the same. 
Similar to the toy model, we list in Table~\ref{tab4} the results for different versions of the algorithm -- original one, modified, mixed and zero-level. 
We also present the results both with and without short-selling constraints. These results show that neither of the versions of the algorithm is performing strictly better than the others. Similar conclusions were observed for various other sets of parameters.

\section{Acceptability Maximization in the Dynamic Setting} \label{sec_dynamic}

In this section, we consider the acceptability maximization problem in a dynamic setting. We use the theory of dynamic coherent acceptability indices introduced in~\cite{BCZ2010} and their link to dynamic coherent risk measures. We briefly recall the key definitions and results from \cite{BCZ2010}, and then focus on acceptability maximization in the context of optimal investment in a multi-period market model. It turns out that the maximal acceptability is constant in a setting when the determining family of dynamic coherent risk measures is recursive and when the underlying market has a recursive structure. We conclude by considering the non-recursive case by focusing on two specific performance measures -- the dynamic risk-adjusted return on capital and the dynamic gain-to-loss ratio -- where we use the specific structure of the problem to introduce a solution scheme tailored to these performance measures.

\subsection{Dynamic CAIs and dynamic CRMs}
The concept of a coherent acceptability index was first extended to a dynamic setting in~\cite{BCZ2010} and consequently studied in \cite{BCDK2013,BCC2014,RosazzaGianinSgarra2012,BiaginiBion-Nadal2012}. 
A dynamic coherent acceptability index (DCAI) is meant to measure the performance of financial positions or instruments over time, accounting for the incoming flow of information. 
We start by briefly recalling the setup of \cite{BCZ2010}, where DCAIs are designed to measure the performance of (discounted) cash flows or dividend streams or unrealized P\&Ls.  Most of the properties from the static setup are naturally transferred to  the dynamic case. An addition is the time consistency property, which stays at the core of financial interpretations of DCAIs, but is also fundamentally used in establishing the dual representations. We refer to \cite{BCP2014,BCP2014a} for an in-depth discussion of various forms of time-consistency in decision making, in particular those arising in the theory of dynamic risk and performance measures. Following \cite{BCZ2010}, we take a discrete and finite state setting by denoting  $\cT := \{0,1,\ldots,T\}$  for some fixed $T\in\bN$, and letting $(\Omega, \mathcal{F}, \bF=(\mathcal{F}_s)_{s \in \cT}, \mathbb{P})$ be  a filtered probability space, with $\mathbb{P}$ having full support. We will write $\bE_t$ instead of the conditional expectation given $\cF_t$. Without loss of generality, we will assume that $\mathcal{F}_0$ is trivial. 
The dividend streams or unrealized P\&Ls will be modeled as $\bF$-adapted real-valued stochastic process $D = \{D_t\}_{t=0}^{T}$.  We will denote by $\bD$ the set of all such processes, and by $L_t = L_t(\Omega, \mathcal{F}_t, \mathbb{P})$ the $\mathcal{F}_t$-measurable random variables. As  usual, for $A\subset\Omega$ will $\1_A$ denote the indicator function which is equal to one for $\omega\in A$ and zero otherwise. 
Without loss of generality, we assume a zero interest rate, or, view $D\in\bD$ as discounted cash flows. Operations between random variables, such as minimum, maximum, product, or sum will be understood $\omega$-wise.

\begin{definition}\label{def_dAI}
A \textbf{dynamic coherent acceptability index (DCAI)}  is a function $\alpha: \cT \times \mathbb{D} \times \Omega \rightarrow [0, \infty]$ satisfying  for all times $t \in \cT,$ all cash flows $D, D' \in \mathbb{D},$ all events $A \in \mathcal{F}_t$, and all random variables $\lambda \in L_t$
\begin{enumerate}
\item[(A1)] Adaptiveness: $\alpha_t (D)$ is $\mathcal{F}_t$-measurable,
\item[(A2)] Independence of the past: if $\1_A D_s = \1_A D_s'$ for all $s \geq t$, then $\1_A \alpha_t(D) = \1_A \alpha_t(D')$,
\item[(A3)] Monotonicity: if $D_s \geq D_s'$ for all $s\geq t$, then $\alpha_t(D) \geq \alpha_t(D')$,
\item[(A4)] Scale invariance: $\alpha_t(\lambda D) = \alpha_t(D)$ for all $\lambda>0$,
\item[(A5)] Quasi-concavity: $\alpha_t(\lambda D + (1-\lambda)D') \geq \min\{ \alpha_t(D), \alpha_t(D') \}$ for $0 \leq \lambda \leq 1$
\item[(A6)] Translation invariance: $\alpha_t(D + m\1_{\{t\}}) = \alpha_t(D + m\1_{\{s\}})$ for any $m \in L_t$ and $s \geq t$, 
\item[(A7)] Dynamic consistency:  if $D_t \geq 0 \geq D_t'$ and there exists an $m \in L_t$ such that $\alpha_{t+1} (D) \geq m \geq \alpha_{t+1} (D')$, then $\alpha_{t} (D) \geq m \geq \alpha_{t} (D')$.
\end{enumerate}
A DCAI $\alpha$ is \textbf{normalized} if for all $t \in \cT, \ \omega \in \Omega$,  there exist  $D, D' \in \mathbb{D}$ such that $\alpha_t (D, \omega) = +\infty$ and $\alpha_t (D', \omega) = 0$. It is \textbf{right-continuous} if $\lim\limits_{c \to 0^+} \alpha_t (D + c \1_{\{t\}}, \omega) = \alpha_t (D, \omega)$ for any $t \in \cT, D \in \mathbb{D}, \omega \in \Omega$.
\end{definition}

As in the static case, DCAIs are closely related to dynamic coherent risk measures (DCRMs). 
\begin{definition}
\label{def_dRM}
A \textbf{dynamic coherent risk measure (DCRM)} is a function $\rho: \cT \times \mathbb{D} \times \Omega \rightarrow  \mathbb{R}$ satisfying for all times $t \in \cT,$ all cash flows $D, D' \in \mathbb{D},$ all states $\omega \in \Omega$, all events $A \in \mathcal{F}_t$, and all random variables $\lambda \in L_t$
\begin{enumerate}
\item[(R1)] Adaptiveness: $\rho_t (D)$ is $\mathcal{F}_t$-measurable,
\item[(R2)] Independence of the past: if $\1_A D_s = \1_A D_s'$ for all $s \geq t$, then $\1_A \rho_t(D) = \1_A \rho_t(D')$,
\item[(R3)] Monotonicity: if $D_s \geq D_s'$ for all $s \geq t$, then $\rho_t(D) \leq \rho_t(D')$,
\item[(R4)] Homogeneity: $\rho_t(\lambda D) = \lambda \rho_t(D)$ for all $\lambda>0$,
\item[(R5)] Subadditivity: $\rho_t(D + D') \leq \rho_t(D) + \rho_t(D')$
\item[(R6)] Translation invariance: $\rho_t(D + m\1_{\{s\}}) = \rho_t(D) - m$ for any $m \in L_t$ and $s \geq t$, 
\item[(R7)] Dynamic consistency: $\1_A \left( \min\limits_{\omega \in A} \rho_{t+1} (D, \omega) - D_t  \right) \leq \1_A \rho_t (D) \leq  \1_A \left( \max\limits_{\omega \in A} \rho_{t+1} (D, \omega) - D_t  \right).$
\end{enumerate}
A family of dynamic coherent risk measures $(\rho^x)_{x \in (0, \infty)}$ is called \textbf{increasing} if $x \geq y >0$ implies $\rho_t^x (D) \geq \rho_t^y (D)$ for any $t \in \cT, D \in \mathbb{D}$. It is \textbf{left-continuous} at $x_0>0$ if $\lim\limits_{x \to x_0^-} \rho^x_t (D, \omega) = \rho^{x_0}_t (D, \omega)$ for any $t \in \cT, D \in \mathbb{D}, \omega \in \Omega$.
\end{definition}

Originally, DCRMs were introduced in~\cite{Riedel2004}, although with a different (stronger) notion of time-consistency, which will be discussed in Section~\ref{sec_recursiveRM}. As  proved in \cite{BCZ2010}, there is a one-to-one relationship between a DCAI and an increasing family of DCRMs, similar to \eqref{AIRM}. Namely, the following assertions hold true:

\begin{enumerate}
\item For a normalized dynamic coherent acceptability index $\alpha$ the functions $(\rho^x)_{x \in (0, \infty)}$ defined as
\begin{align}
\label{AItoRM}
\rho_t^x (D, \omega) = \inf \{ c \in \mathbb{R} : \alpha_t(D + c\1_{\{t\}}, \omega) \geq x \}, .
\end{align}
form an increasing, left-continuous family of dynamic coherent risk measures.
\item For an increasing family of dynamic coherent risk measures $(\rho^x)_{x \in (0, \infty)}$ a function $\alpha$ defined as
\begin{equation}\label{RMtoAI}
\alpha_t (D, \omega) = \sup \{ x \in (0, \infty) : \rho_t^x (D, \omega) \leq 0 \}
\end{equation}
is a normalized, right-continuous dynamic coherent acceptability index. Moreover, there exists an increasing  sequence of sets of probability measures $\set{\cQ_t^x}_{t\in\cT, x\geq 0}$ such that 
\begin{equation}\label{eq:rep3}
\rho_t^x(D, \omega) = -\sup_{\bQ\in\cQ_t^x} \bE^{\bQ}_t \left[\sum_{s=t}^T D_s \right], \quad x>0, \ t\in\cT. 
\end{equation}
The converse implication is also true, under an additional technical property of time consistency of $\set{\cQ_t^x}_{t\in\cT}$, 
\item If $\alpha$ is a normalized, right-continuous dynamic coherent acceptability index, then there exists an increasing, left-continuous family of dynamic coherent risk measures $(\rho^x)_{x\in (0, \infty)}$, such that representation~\eqref{RMtoAI} holds. Vice versa, for an increasing, left-continuous family of dynamic coherent risk measures $(\rho^x)_{x\in (0, \infty)}$ there exists a normalized, right-continuous dynamic coherent acceptability index $\alpha$, such that~\eqref{AItoRM} holds.
\end{enumerate}

Similar to the static case, we are interested in finding the position with highest degree of acceptability. Given a set of available, or feasible, cash flows $\mathcal{D} \subseteq \mathbb{D}$, the problem of interest is
\begin{align}
\max\limits_{D \in \mathcal{D}} \alpha_0 (D). \label{DmaxAI}
\end{align}
It is straightforward to adapt Algorithm~\ref{alg} to the dynamic setup to solve the corresponding version of \eqref{DmaxAI} for a given $t\in\cT$ and $\omega\in\Omega$. However, this approach, generally speaking, is computationally not feasible. Usually one would aim to establish a recursive set of equations in the form of a dynamic programming  principle or Bellman's principle of optimality that would solve \eqref{DmaxAI}, which will be provided in the next section for the optimal investment problem.

\subsection{Optimal portfolio selection problem}
\label{sec_self-fin}

In this section, we consider the acceptability maximization problem in the context of optimal portfolio selection in a market model with  $d$ available assets.  We denote by $R_{s+1}=(R_{s+1}^1, \ldots, R_{s+1}^d)$ the vector of assets (total or gross) returns between time $s$ and time $s+1$, namely, if $S_s^j$ denotes the price of the $j$-th asset at time $s$, then $R_{s+1}^j := \frac{S_{s+1}^j}{S_s^j}$. We assume that $R_1, \dots, R_T$  are  independent and identically distributed on a probability space $(\Omega,\cF,\bP)$, and denote by $(\mathcal{F}_s)_{s \in \cT}$  the natural filtration generated by the process $(R_s)_{s=1,\ldots,T}$. In addition, we assume that all one step asset returns $R_s^j$ are strictly positive. Note that we implicitly assume that these assets do not pay dividends.

We assume that the investor starts with a positive initial wealth $V_0>0$, and invests it in the $d$ available assets by following a self-financing trading strategy, possibly with some additional trading constraints. A trading strategy is an adapted stochastic process $h=(h_s)_{s=0,\ldots,T-1}$ with $h_s = (h_s^1,\ldots,h_s^d)$, where $h_s^i$ is the monetary amount invested in asset $i$ between time $s$ and $s+1$.   
The portfolio value at time $s+1$ arising from the trading strategy $h$ is given by $V_{s+1} (h) = \trans{R_{s+1}} h_s$  for any $s = 0, \dots, T-1$.
We consider two feasible sets in particular, one with no trading constraints and one with short-selling constraints.
The set of all self-financing trading strategies with initial value $V_0$ is 
$$
\mathcal{H}_0 (V_0) := \{ (h_s)_{s = 0, \dots, T-1} \; \vert \; \trans{\mathbf{1}} h_s = V_s, V_{s+1} = \trans{R_{s+1}} h_s, s = 0, \dots, T-1 \}.
$$
Correspondingly, the set of feasible trading strategies with short-selling constraints is 
$$
\mathcal{H}_0^+ (V_0) = \{ (h_s)_{s = 0, \dots, T-1} \; \vert \;  h\in \cH_0(V_0), \  h_s^j \geq 0, s = 0, \dots, T-1, \ j=1,\ldots,d \}.
$$
The time $t$ feasible sets $\cH_t(V_t)$ and $\cH_t^+(V_t)$ are defined analogously. The next result shows that the feasible sets are positive homogeneous and recursive. 

\begin{restatable}{lemma}{lemmaPhi}
\label{lemma_W}
\begin{enumerate}
\item For a positive $\mathcal{F}_t$-measurable wealth $V_t$ the feasible sets scale as follows
\begin{align*}
 \mathcal{H}_t (V_t) = V_t \cdot \mathcal{H}_t (1)    \; \text{ and } \; \mathcal{H}_t^+ (V_t) = V_t \cdot \mathcal{H}_t^+(1).
\end{align*}

\item The feasible sets are recursive,
\begin{align*}
\mathcal{H}_t (V_t) &= \left\lbrace (h_s)_{s = t, \dots, T-1} \; \vert \; h_t \in H_t (V_t), (h_s)_{s = t+1, \dots, T-1} \in \mathcal{H}_{t+1} (\trans{R_{t+1}} h_t) \right\rbrace, \\
\mathcal{H}_t^+ (V_t) &= \left\lbrace (h_s)_{s = t, \dots, T-1} \; \vert \; h_t \in H_t^+ (V_t), (h_s)_{s = t+1, \dots, T-1} \in \mathcal{H}_{t+1}^+ (\trans{R_{t+1}} h_t) \right\rbrace, 
\end{align*}
where  $H_t (V_t) = \{ h_t \; \vert \; \trans{\mathbf{1}} h_t = V_t \}$ and $H_t^+ (V_t) = \{ h_t \in H_t(V_t) \; \vert \;  h_t \geq 0 \}$. 
\end{enumerate}
\end{restatable}
\begin{proof}
The proof is deferred to Appendix~\ref{AppendixB}.
\end{proof}

Our aim is to find the optimal trading strategy among the feasible ones by maximizing the portfolio's acceptability as measured by a given DCAI $\alpha$. We recall that the dynamic setup of \cite{BCZ2010} and \cite{Riedel2004} assumes that the inputs $D$ to a DCAI are (discounted) dividend processes, a setup usually convenient for pricing purposes or assessing the performance or riskiness of some dividend paying securities, or random future cash-flows (cf. \cite{AcciaioFollmerPenner2010,AcciaioPenner2010,BCIR2012,BCC2014} and references therein). When dealing with optimal investment (i.e. an optimal portfolio selection problem), traditionally and also more conveniently, one works with the value process or the (discounted) cumulative dividend process. 
Given a portfolio with value process $V = (V_s)_{s=0, \dots, T}$, the corresponding dividend stream $D = (D_s)_{s=0, \dots, T}$ is defined as
$$
D_s = V_s - V_{s-1}, \quad s = 1, \dots, T,
$$
and $D_0 = 0$. Thus, the cumulative P\&L up to time $t$ becomes $\sum_{s=0}^{t}D_s = V_t-V_0$. 
We refer the reader to \cite{AcciaioPenner2010,BCDK2013} for a detailed discussion on use of dividend streams and cumulative dividend streams within the general theory of assessment indices. 

We denote by $V(h)$ the wealth process generated by the trading strategy $h$ and $D(h)$ will stand for the corresponding dividend stream. In addition, for a given dividend stream $D = (D_0, \dots, D_T)$ we define the time $t$ tail dividend stream as $D^{[t, T]} := (0, \dots, 0, D_{t}, \dots, D_T)$, and we also put $D^{[t]} := (0, \dots, 0, D_{t}, 0,  \dots, 0)$.

The optimization problem we wish to solve at initial time is
\begin{align}
\label{ProbA}
\tag{$A_0 (V_0)$}
\max\limits_{h \in \mathcal{H}_0^+ (V_0)} \alpha_0 (D(h))
\end{align}
or the variant thereof, if short-selling is allowed, in which case the feasible set is $\mathcal{H}_0(V_0)$. By property (A2), independence of the past of DCAIs, we have that 
$\alpha_t (D) = \alpha_t (D^{[t, T]})$.
This would suggest that in order to solve~\ref{ProbA} we should consider the problems 
\begin{align}
\label{ProbAtx}
\tag{$A_t (V_t)$}
\max\limits_{h \in \mathcal{H}_t^+ (V_t)} \alpha_t (D^{[t, T]} (h)).
\end{align}
One certainly can study  \eqref{ProbAtx}, and try to establish a dynamic programming principle for this stochastic control problem, although, generally speaking, this problem is not time consistent. 
More importantly, from a practical point of view, including the change in portfolio value from time $t-1$ to $t$, namely the term $D_t=V_t-V_{t-1}$,  in the optimization criteria at time $t$ is less desirable. In particular, using \eqref{eq:rep3}, one would optimize at time $t$ a function that depends on $V_T-V_{t-1}$, rather than a function depending on total future return $V_T-V_t$. With this in mind, we introduce and focus our attention on a family of auxiliary acceptability maximization problems 
\begin{align}
\label{ProbAt}
\tag{$\tilde{A}_t (V_t)$}
\max\limits_{h \in \mathcal{H}_t^+ (V_t)} \alpha_t (D^{[t+1, T]} (h)), 
\end{align}
which are more in line with the setup from the optimal portfolio selection problem. 
Note that for any trading strategy $h$ the cash flows $D (h)$ and $D^{[1, T]} (h)$ coincide, and hence at the initial time the auxiliary problem $\tilde{A}_0 (V_0)$ is the same as the original problem~\ref{ProbA}. Therefore, solving the auxiliary family of problems, which we will address next, will lead to the solution of the original acceptability maximization problem.

\subsubsection{The Case of Recursive Risk Measures $(\rho_t^x)_{t\in\cT}$} \label{sec_recursiveRM}
As we already mentioned, the form of the time consistency property (R7) for DCRMs is tailored for the robust representation \eqref{RMtoAI} of DCAIs with the time consistency property (A7). This form of time consistency is weaker than the so-called strong time consistency of risk measures:  
\begin{itemize}
\item[(R7')] \textit{Strong time consistency:} for any $D,D'\in\bD$ and $t=0,\ldots,T-1$, if $D_t = D_t'$ and $\rho_{t+1} (D) = \rho_{t+1} (D'),$ then $\rho_{t} (D) = \rho_{t} (D')$. 
\end{itemize}
Strong time consistency (R7') is the one usually  associated with dynamic risk measures (cf. \cite{Riedel2004,BCP2014}), due to its natural financial interpretation, but also because of its equivalence to:
\begin{itemize}
	\item[(R7'')] \textit{Recursiveness:} $\rho_{t} (D) = \rho_t \left( -\rho_{t+1}(D)\1_{\{t+1\}} \right) - D_t$, for any $D\in\bD$ and $t=0,\ldots,T-1$.  	
\end{itemize}
One major benefit of having the recursive property is its direct applicability to stochastic control problems. This very property makes many stochastic control problems with risk as the terminal criteria to be time consistent. Note that such recursive property in principle can not be satisfied by any DCAI, cf. \cite{BCP2014a,BCP2014}.  

In this section, we study the acceptability maximization problem \eqref{ProbAt} assuming that the corresponding family of risk measures is strongly time consistent.  We work under the market setup of Section~\ref{sec_self-fin} with short-selling constraints and with an initial value $V_0>0$. 

The one-step risk measures generated by $\rho^x$ are defined as 
$$
\rho_{t,t+1}^x(Z,\omega):= \rho_t^x(\mathbf{0} + \1_{\{t+1\}}Z)(\omega), \quad t=0,\ldots,T-1, \ \omega\in\Omega,
$$
for any $\cF_{t+1}$-measurable random variable $Z$, here $\mathbf{0}$ denotes the zero process.  In what follows we assume that the one-step risk measures are identical across all nodes of the multinomial model. Namely, with $\cP_t$ denoting the partition of $\Omega$ that generates $\mathcal{F}_t$, we assume that for any $t,s\in\cT$, and any $\Omega_t\in\cP_t$ and $\Omega_s\in\cP_s$
\begin{equation}\label{eq:1stepRisk}
\rho_{t,t+1}^x(D_{t+1},\omega)  = \rho_{s,s+1}^x(D'_{s+1},\omega')
\end{equation} 
for all $\omega \in \Omega_t$, all $\omega' \in \Omega_s$ and all $D, D'$ satisfying $\1_{{\Omega}_t} D_{t+1} \overset{(d)}{=} \1_{\Omega_s} D'_{s+1} $ and zero otherwise. 
As previously, we denote the maximal acceptability attainable at the market as
\begin{align*}
\alpha^*_t (V_t; \omega) := \sup\limits_{h \in \mathcal{H}_t^+ (V_t)} \alpha_t (D^{[t+1, T]} (h); \omega).
\end{align*}

Under the above, what may appear, natural assumptions, we obtain a somehow surprising result: the maximal acceptability $\alpha^*$ is constant across wealth level, time and states of the world. 

\begin{theorem} \label{thm_alpha}
Let $\alpha$ be a normalized right-continuous DCAI and $(\rho^x)_{x \in (0, \infty)}$ be the corresponding family of DCRMs. Assume that for each $x>0$ the DCRM $\rho^x$ is strongly time consistent, and all the one step risk measures $\rho_{t,t+1}^x$ satisfy \eqref{eq:1stepRisk}.
Then, under the market model assumption of this section, the maximal acceptability $\alpha^*_t$ is independent of the wealth, time and state, that is, 
\begin{align*}
\alpha^*_t (V_t;\omega) = \alpha^*_0 (1),
\end{align*}
for all $t\in\cT$, $\omega \in \Omega$ and  positive $V_t \in L_t$.
\end{theorem}
\begin{proof}
The proof is given in Appendix~\ref{AppendixB}.
\end{proof}

In view of Theorem~\ref{thm_alpha} the auxiliary acceptability maximization has a constant optimal objective value in time, and since at the initial time the auxiliary and the original problem coincide, we obtain that it suffices to solve $\tilde{A}_{T-1} (1)(\bar{\omega})$ for some $\bar{\omega} \in \Omega$ instead of~\ref{ProbA}. The next result shows how to construct the corresponding optimal trading strategy. 

\begin{theorem}\label{thm_alpha_solution}
Assume that for some $\bar{\omega} \in \Omega$ the supremum $\alpha^*_{T-1} (1; \bar{\omega})$ is attained and denote by $h^* \in \mathbb{R}^d$ the corresponding optimal position (given $\bar{\omega}$),
\begin{align*}
h^* = \arg\max\limits_{h_{T-1} \in H_{T-1}^+ (1)} \alpha_{T-1} (D^{[T]} (h_{T-1}); \bar{\omega}) (\bar{\omega}),
\end{align*}
where $H_{T-1}^+ (1)$ was defined in Lemma~\ref{lemma_W}(2).
Let  $(\bar{h}_s)_{s=0, \dots, T-1}$ be the trading strategy defined as
\begin{align*}
\bar{h}_0 &= V_0 \cdot h^*, \\
\bar{h}_s &= V_{s}(\bar{h}_{s-1}) \cdot h^*,  \quad  s=1, \dots, T-1.
\end{align*}
Then, the trading strategy $\bar{h}$ is an optimal solution of~\ref{ProbA}, i.e.~$\alpha_0 (D(\bar{h})) = \alpha^*$.
\end{theorem}
\begin{proof}
The proof is given in Appendix~\ref{AppendixB}.
\end{proof}

The results of this subsection rely on the properties of the family of risk measures $\set{(\rho^x_t)_{t=0, \dots, T}}_{x>0}$ corresponding to the DCAI $\alpha$ under consideration. Similar to the static case, one may be interested in the risk minimization problem corresponding to a fixed level $x>0$. As may be expected, the recursive setting of this section has direct implications on the minimal achievable risk (infimum of the risk minimization problem). It can be proved that the minimal risk is positively homogeneous and it also has a recursive form. Unlike the maximal acceptability it is not constant, but it maintains the same sign over all times and states. Furthermore, if an optimal solution (optimal trading strategy) exists, it can be constructed recursively in the spirit of Theorem~\ref{thm_alpha_solution}.

\subsubsection{The Case of Dynamic RAROC}\label{sec:dRAROC}
Using the definition of the static RAROC, the identity \eqref{eq:RAROC}, as well as the representation \eqref{eq:rep3}, one naturally defines the dynamic risk adjusted return on capital (dRAROC) as follows:
\begin{align*}
\text{dRAROC}_t (D) = 
\dfrac{\left( \mathbb{E}_t \left( \sum_{s=t}^T D_s  \right) \right)^+ }{\left( \pi_t \left( \sum_{s=t}^T D_s \right) \right)^+ }, \quad D\in\bD,
\end{align*}
with the convention $\frac{a}{0} = +\infty$, where $\pi$ is a given dynamic coherent risk measure (not to be confused with the family of DCRMs corresponding to an acceptability index).

As was shown in \cite[Section~6]{BCZ2010}, dRAROC fulfills the properties (A1)-(A6), but it, in general, fails to satisfy the dynamic consistency property (A7), and therefore, it is not a DCAI. Nevertheless, for some choices of $\pi$, dRAROC satisfies some weaker forms of time consistency. In particular, if $\pi$ is the dynamic version  of $\tvar$, then the corresponding dRAROC is so-called  semi-weakly acceptance time consistent, but not  semi-weakly rejection time consistent; for more details see \cite{BCP2014,BCP2014a}. This will be the example we consider in our numerical experiment below. With this in mind, we are interested in identifying an investment with the highest performance measured by dRAROC, i.e. the maximization of the dRAROC-performance in the framework of self-financing portfolios introduced in  Section~\ref{sec_self-fin}. Furthermore, we focus our attention only on the case of a feasible set with short-selling constraints $\cH^+_0 (V_0)$, but most of the results can be extended to the case with no trading constraints. Hence, we wish to solve the following optimization problem: 
\begin{align}\label{prob_dRAROC}
\max\limits_{h \in \mathcal{H}_0^+ (V_0)} \textrm{dRAROC}_0 (D (h)).
\end{align}
As already mentioned, this problem is time-inconsistent (in the sense of optimal control), and in view of the above it does not fit the framework of Section~\ref{sec_recursiveRM}. 

We will take the approach of \cite{KovacovaRudloff2019} to deal with time-inconsistency of \eqref{prob_dRAROC}. 
First we note that for a positive level $x > 0$ 
\begin{align*}
\textrm{dRAROC}_0 (D) \geq x \; \Leftrightarrow \; \min \left\lbrace \pi_0 \left( \sum_{s=0}^T D_s \right),  - \mathbb{E}_0 \left( \sum_{s=0}^T D_s \right) + x \pi_0 \left( \sum_{s=0}^T D_s \right) \right\rbrace \leq 0.
\end{align*}
Using this, one could apply the idea of Algorithm~\ref{alg} to the family  of functions $\rho^x_0 (\cdot) = \min \{ \pi_0 (\cdot),  -\mathbb{E}_0 (\cdot) + x \pi_0 (\cdot) \}$ for $x>0$. 
In the nutshell, the procedure would consist of the following: First, minimize the risk $\pi_0 (\cdot)$ among the feasible positions, i.e.~solve the mean-risk problem with an infinite risk aversion. If the optimal value were negative, an infinite performance measured by $\textrm{dRAROC}_0$ would be implied. Second, repeatedly minimize $-\mathbb{E}_0 (\cdot) + x \pi_0 (\cdot)$ among the feasible positions for various levels $x$ -- that is, solve the mean-risk problem for the risk aversion at various levels $x$.
Therefore, the algorithm would be iteratively computing elements of the mean-risk efficient frontier. If the (full) efficient frontier, i.e. the set of all portfolios that are not dominated in terms of their mean and risk, was available instead, the optimal solution of~\eqref{prob_dRAROC} could be found simply as the element of the frontier with the highest ratio of the mean to the risk. Of course, applying Algorithm~\ref{alg} is not computationally efficient in the dynamic setup, but it  motivates us to compute the efficient frontier, i.e. to consider
the bi-objective mean-risk problem
\begin{align}
\label{prob_mean-risk}
\min\limits_{h \in \mathcal{H}_0^+ (V_0)} 
\begin{pmatrix}
-&\mathbb{E}_0 (V_T (h) - V_0) \\
&\pi_0 (V_T (h) - V_0)
\end{pmatrix} \text{ w.r.t. } \leq_{\mathbb{R}^2_+},
\end{align}
where we used the fact that  $ \sum_{s=0}^T D_s (h) = V_T (h) - V_0$. This will also overcome the problem of time-inconsistency of~\eqref{prob_dRAROC} and thus lead to an efficient way to solve~\eqref{prob_dRAROC}. As it turns out, problem~\eqref{prob_mean-risk} is time consistent in the set-valued sense, i.e.,
the set-valued Bellman's principle of optimality recently proposed in~\cite{KovacovaRudloff2019} provides a way to solve the mean-risk problem~\eqref{prob_mean-risk} recursively, assuming that the dynamic risk measure $\pi$ is recursive, i.e. strongly time consistent. We emphasis that the recursiveness of $\pi$ does not imply the recursiveness of the members of the family $(\rho^x_0)_{x \in (0, \infty)}$ and is a separate property from the dynamic consistency of the performance measure $\textrm{dRAROC}$.  
The set-valued Bellman's principle of optimality of~\cite{KovacovaRudloff2019} also provides the intermediate mean-risk efficient frontiers, namely it solves the sequence of mean-risk problems
\begin{align*}
\min\limits_{h \in \mathcal{H}_t^+ (V_t)} 
\begin{pmatrix}
-&\mathbb{E}_t (V_T (h) - V_t) \\
&\pi_t (V_T (h) - V_t)
\end{pmatrix} \text{ w.r.t. } \leq_{L_t(\mathbb{R}^2_+)},
\end{align*}
for each time point $t =0, \dots, T-1$. Note that since $\textrm{dRAROC}_0$ is equal to the ratio of the mean to the risk, the element of the frontier of the time-consistent problem~\eqref{prob_mean-risk} with the highest ratio is the optimal solution of the time-inconsistent problem~\eqref{prob_dRAROC}. The same can be said about the intermediate mean-risk efficient frontiers and (auxiliary) problems $\max\limits_{h \in \mathcal{H}_t^+ (V_t)} \textrm{dRAROC}_t ( D^{[t+1, T]} (h))$.

We illustrate this on a dynamic version of the example from Section~\ref{sec_examples}. We consider the  market model with two assets, and with one-time-step asset returns $R^i_t, \ i=1,2$, having the probability law given in Panel A of Table~\ref{tab1}, and we take $T=6$. We take the DCRM $\pi$ to be the recursive dynamic  $\tvar$ at significance level $1\%$. We recall that the dynamic $\tvar$ is defined analogously to the static $\tvar$ by replacing $\var$ with the conditional $\var$, which in turn is defined as a conditional quantile.

In Figure~\ref{fig1} we display the mean-risk efficient frontier of problem~\eqref{prob_mean-risk}, as well as the intermediate frontiers. The bright green points correspond to the elements with the highest dRAROC.

\begin{figure}[h]
\centering
\includegraphics[width=\textwidth]{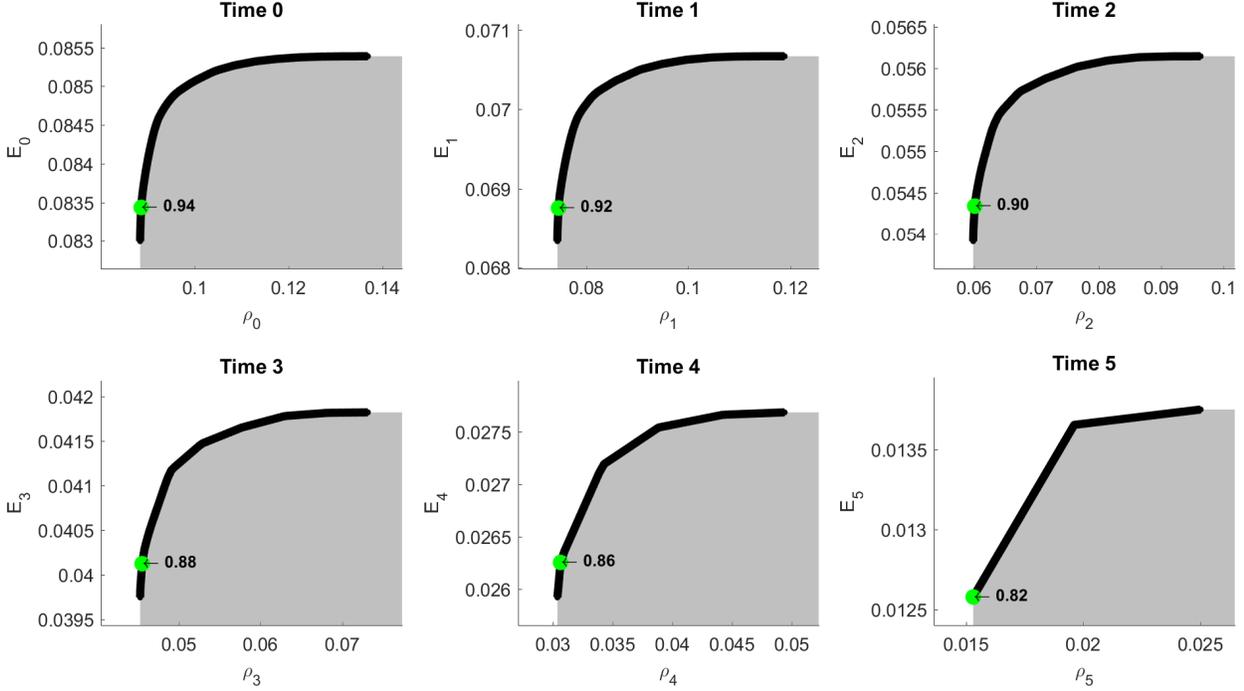}
\caption{Efficient frontiers (black) of the mean-risk problems and elements with the highest mean-to-risk ratio (green). All frontiers are depicted in the $(\rho, \mathbb{E})$ plane for the returns $v_T - v_t$ with $v_t = 1$.}
\label{fig1}
\end{figure}

The trading strategy $(h_t)_{t = 0, \dots, T-1}$ corresponding to the highest-$\textrm{dRAROC}_0$ element of the time $0$ frontier can be recovered from the solution of the mean-risk problem, see~\cite{KovacovaRudloff2019} for details. The mean-risk profiles, and the corresponding values of dRAROC, of this portfolio in the subsequent time points are determined by the strategy itself and vary over times and states of the world. They are depicted as yellow triangles for a selected state of the world $\omega$ in Figure~\ref{fig2}.
All of them lie on the efficient frontiers (yellow triangles), but, in general, do not coincide with the highest-$\textrm{dRAROC}_t$ element (bright green points). This confirms the time-inconsistency of dRAROC -- the strategy optimal from the viewpoint of time $t=0$ is not dRAROC-maximal at the subsequent time instances. 

For comparison, we include also a myopic (magenta square) and an inconsistent switching (red diamond) approach. In the myopic case, the investor at each time solves a one step optimization problem, hence looking always only one period ahead and chooses the position that maximizes the RAROC over this one-period horizon. The switching strategy represents a time inconsistent behavior in the sense that at time $t$  the $\textrm{dRAROC}$-maximal element of the (time consistent) frontier is selected, the  trading strategy $(h_s^{(t)})_{s = t, \dots, T-1}$ corresponding to it is found, and the position $h_t^{(t)}$ is taken. At the next time $t+1$  the previously found trading strategy $(h_s^{(t)})_{s = t+1, \dots, T-1}$ is discarded and a new one, $(h_s^{(t+1)})_{s = t+1, \dots, T-1}$ corresponding to the dRAROC-maximal element of the $t+1$ frontier, is selected. Since each (efficient) trading strategy is discarded after one time period, none of the corresponding (dRAROC-optimal) mean-risk profiles are ever realized. Figure~\ref{fig2} shows the actual means, risks and values of dRAROC that these behaviors yield. Clearly, neither the myopic nor the switching give at any time (except at $T-1$) the maximal performance. They even lead to portfolios, which are not mean-risk efficient at all, i.e. they do not lie on the frontier.

\begin{figure}[h]
\centering
\includegraphics[width=\textwidth]{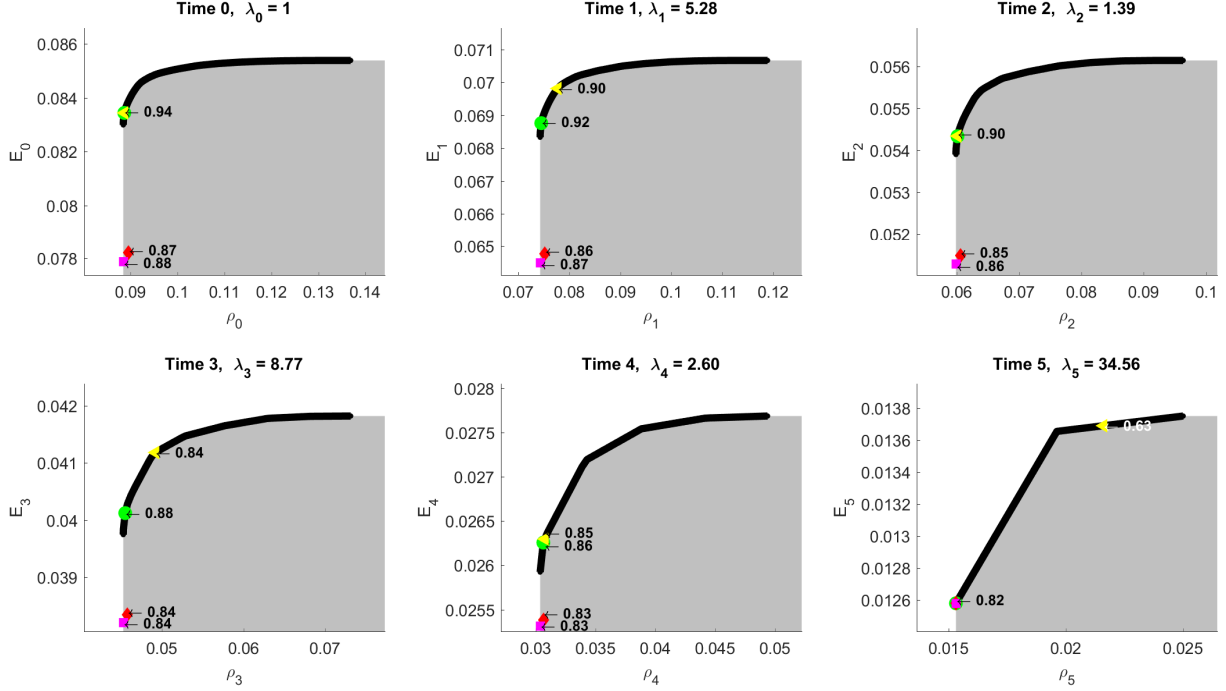}
\caption{Efficient frontiers for returns over time. The mean-risk profiles and the corresponding values of dRAROC are depicted for three trading strategies: the time consistent mean-risk strategy in one state $\omega$ (yellow triangle), the switching strategy (red diamond) and the myopic strategy (magenta square). The element of the frontier with the highest dRAROC is also depicted at each time (green circle).}
\label{fig2}
\end{figure}

Finally, let us look again at the strategy depicted in yellow, namely the strategy that solves \eqref{prob_dRAROC} at time zero. While this stochastic control problem is time-inconsistent, one can ask which objective does the optimal strategy maximize at the intermediate times. Note that the $\textrm{dRAROC}_0$-maximal element of the time $0$ frontier corresponds to a nonlinear scalarization
\begin{align*}
\begin{pmatrix}
-&\mathbb{E}_0 (V_T (h) - V_0) \\
&\rho_0 (V_T (h) - V_0)
\end{pmatrix} \; \mapsto \; \dfrac{\mathbb{E}_0 (V_T (h) - V_0)}{\rho_0 (V_T (h) - V_0)}.
\end{align*}
Thus, we will concentrate on a class of non-linear scalarizations including the one above. Specifically, we consider scalarizations of the time $t$ frontier of the form 
\begin{align}\label{eq_ratio}
\begin{pmatrix}
-&\mathbb{E}_t (V_T (h) - V_t) \\
&\rho_t (V_T (h) - V_t)
\end{pmatrix}
\; \mapsto \; 
\dfrac{\left( \mathbb{E}_t (V_T (h) - V_t) \right)^{\lambda_t}}{\rho_t (V_T (h) - V_t)},
\end{align}
where the mapping is fully determined by the value $\lambda_t$, which can be interpreted as a non-linear risk aversion parameter. For any given efficient trading strategy, one can compute the value of $\lambda_t$, such that the strategy is an optimal solution of a scalar problem with the objective~\eqref{eq_ratio}. This way, a sequence of $\lambda_0, \dots, \lambda_{T-1}$ can be computed for the $\textrm{dRAROC}_0$--optimal strategy $ (h_t)_{t=0, \dots, T-1}$ (represented on the frontiers by the yellow mean-risk pairs). Since the frontiers (and the mean-risk profiles) are adapted, also the corresponding scalarization coefficient $\lambda_t$ will be adapted. We computed the corresponding $\lambda_t$ in the given state of the world $\omega$ and depicted it also in Figure~\ref{fig2}.

Thus, the sequence of scalar problems~\eqref{eq_ratio} is time consistent in the usual sense for the computed risk aversion parameters $\lambda_0, \dots, \lambda_{T-1}$.
As $\lambda_0=1$  is by construction included, a time zero member of this time consistent family is the  $\textrm{dRAROC}_0$--maximization problem. Thus, an investor with a $\textrm{dRAROC}_0$ criteria at time zero and a $\textrm{dRAROC}_t$ like criteria, that differs only in a changed risk aversion parameter $\lambda_t$, where $\lambda_t$ is changing in a certain manner according to the changes in the stock market, would behave time consistent in the classical sense. This is in line with the findings about the moving scalarization (a time and state dependent risk aversion parameter) that leads to a time consistent problem and a time consistent behaviour of the investor as also discussed in the mean-risk portfolio optimization problem in \cite{KovacovaRudloff2019} and for other otherwise time inconsistent problems in~\cite{KarnamETAL17}.

\subsubsection{The Case of Dynamic GLR}\label{sec:dGLR}
Similar to dRAROC, the dynamic gain-to-loss ratio (dGLR) is defined as 
\begin{align}
\label{def_dGLR}
\text{dGLR}_t (D) = 
\dfrac{\left(\mathbb{E}_t \left( \sum_{s=t}^T D_s \right)\right)^+ }{\mathbb{E}_t \left( \left( \sum_{s=t}^T D_s \right)^- \right) }, \quad D\in\bD,
\end{align}
with the convention $\frac{a}{0} := +\infty.$ Unlike dRAROC, dGLR is a normalized and right-continuous DCAI (see \cite[Section~6]{BCZ2010}). Our aim is to identify among all self-financing portfolios the ones with the highest dGLR, that is to solve the problem
\begin{align}\label{prob_dGLR}
\max\limits_{h \in \mathcal{H}_0 (V_0)} \textrm{dGLR}_0 (D (h)).
\end{align}
Similar to the static GLR, the family $\rho^x$ of DCRM from the robust representation is identified by the conditional expectiles, and since the conditional expectiles are not strongly time consistent, the results of Subsection~\ref{sec_recursiveRM} do not apply here. As was also noted in the static case, instead of the corresponding family of risk measures one can consider the family $-\mathbb{E}_0 (\cdot) + x  \mathbb{E}_0 \left( \left( \cdot \right)^- \right)$ for $x>0$. Note that for any fixed time instance one can view the problem as a static one, and thus one can apply Algorithm~\ref{alg}, but this would be, as discussed before, computationally infeasible to do for all $t\in\cT$.  Here, with the intention of obtaining a Bellman's principle of optimality,  we take an approach inspired by the previous subsection and in the spirit of \cite{KovacovaRudloff2019}. Motivated by the numerator and denominator of~\eqref{def_dGLR}, we consider the bi-objective mean-loss problem
\begin{align}
\label{prob_mean-loss}
\min\limits_{h \in \mathcal{H}_0 (V_0)} 
\begin{pmatrix}
&-\mathbb{E}_0 (V_T(h) - V_0) \\
&\mathbb{E}_0 \left( \left( V_T(h) - V_0 \right)^- \right)
\end{pmatrix} \text{ w.r.t. } \leq_{\mathbb{R}^2_+}.
\end{align}
By the same argument as in the dRAROC case, the element of the efficient frontier with the highest ratio corresponds to the portfolio with the highest value of $\textrm{dGLR}_0$.
The recursive approach of~\cite{KovacovaRudloff2019}, unfortunately, can not be applied directly here, due to the lack of translation invariance of the objective function $\mathbb{E}_t \left(  X^- \right),$ which makes it impossible to express  $\mathbb{E}_0 \left( \left( V_T(h) - V_0 \right)^- \right)$  through $\mathbb{E}_t \left( \left( V_T(h) - V_t \right)^- \right)$. Nevertheless, to solve \eqref{prob_mean-loss}, we consider the following sequence of bi-objective problems
\begin{align}
\label{prob_mean-losst}
\min\limits_{h \in \mathcal{H}_t (V_t)} 
\begin{pmatrix}
&-\mathbb{E}_t (V_T(h) - V_0) \\
&\mathbb{E}_t \left( \left( V_T(h) - V_0 \right)^- \right)
\end{pmatrix} \text{ w.r.t. } \leq_{\mathbb{R}^2_+}\,,
\end{align}
where $V_0$ is the fixed initial wealth. Problem~\eqref{prob_mean-losst} does not have a natural interpretation as a mean-loss problem, unless $V_t = V_0$, however, it does give a recursive solution of~\eqref{prob_mean-loss} in terms of the set-valued Bellman's principle of~\cite{KovacovaRudloff2019}. 

We also note that the computational approach from~\cite{KovacovaRudloff2019} based on scaling arguments is not applicable here either, and therefore one needs to solve \eqref{prob_mean-losst} for any $V_t.$ As the problems~\eqref{prob_mean-losst} can be rewritten as bi-objective linear optimization problems and differ only in the right-hand side of the constraints, they form a class of parametric bi-objective linear problems with the parameter $V_t.$ We solved these parametric problems via polyhedral projection (cf. \cite{LoehneWeissing2016}).

We conclude this section by illustrating the solution to \eqref{prob_dGLR} in the same market model setup as in Section~\ref{sec:dRAROC} and by taking the initial wealth $V_0 = 0$. 
 Figure~\ref{fig3} contains the efficient frontier at time $t=0$ and the highest value of problem~\eqref{prob_dGLR} given by $\textrm{dGLR}_0=0.27$. As the intermediate frontiers are computed for all possible values of $V_t$, we depict for illustration only those frontiers corresponding to $V_t = V_0$ for each time point $t$.
 The case of the current wealth $V_t$ coinciding with the initial wealth $V_0$ would give problem~\eqref{prob_mean-losst} the interpretation as the mean-loss problem. Therefore the corresponding maximal value of $\textrm{dGLR}_t$ can be obtained. Since the zero-cost trading strategy can be scaled, the frontier is naturally a half-line. The highest value of dGLR corresponds to the slope of the frontier. The optimal trading strategy of~\eqref{prob_dGLR} can be deduced from the solution of~\eqref{prob_mean-losst}. Thus, an auxiliary, but time-consistent bi-objective problem~\eqref{prob_mean-losst} (following a backward recursion by the set-valued Bellman's principle of optimality) is used to compute the optimal solution of the time-inconsistent problem~\eqref{prob_dGLR}.

\begin{figure}[h]
\centering
\includegraphics[width=\textwidth]{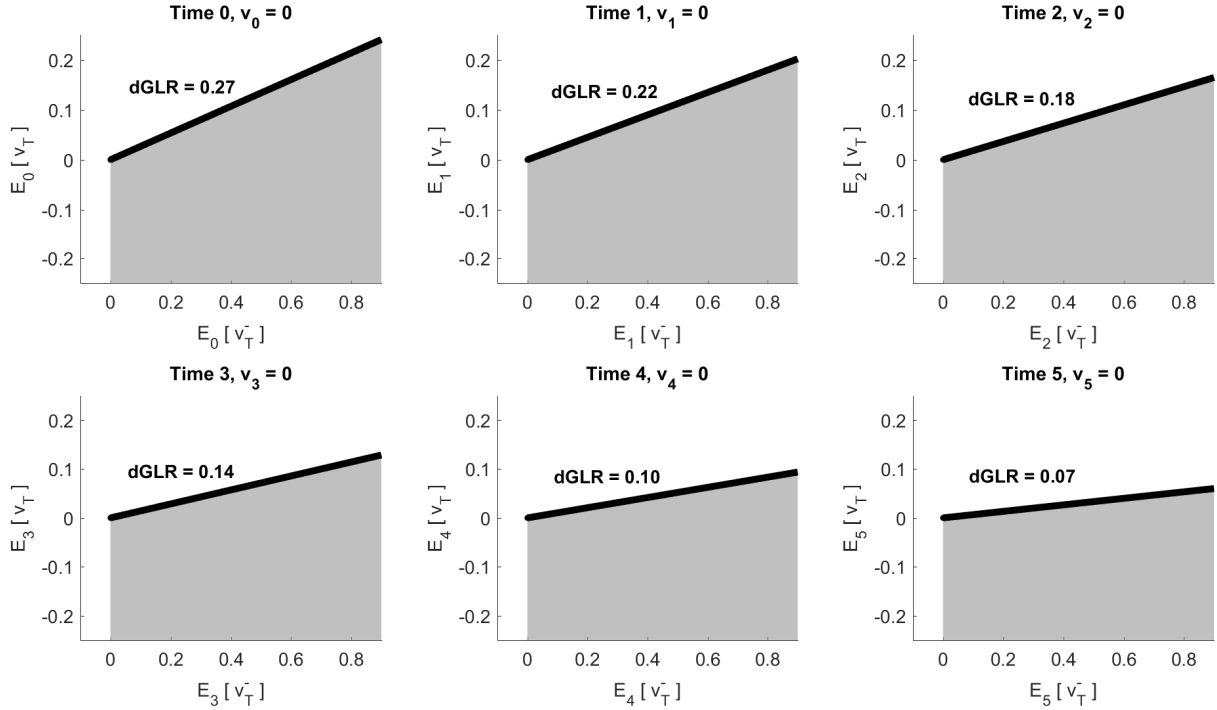}
\caption{Efficient frontiers (black) of the problems~\eqref{prob_mean-losst} depicted for wealth $V_t = 0.$ All frontiers are depicted in the $(\mathbb{E}_t(V_T^-), \mathbb{E}_t (V_T))$ plane. The corresponding highest value of $dGLR$ (the slope of the frontier) is given.}
\label{fig3}
\end{figure}

\section*{Acknowledgments}
IC acknowledges partial support from the National Science Foundation (US) grant DMS-1907568.
IC thanks the Vienna University of Economics and Business for the hospitality during a research visit  related to the workshop on dynamic multivariate programming in March 2018, where this project has been initiated.  The authors thank Christian Diem for helpful discussions in the early stages of the project.

\appendix
\section{Proofs from Section~\ref{sec_static}}
\label{AppendixA}
\begin{proof}[Proof of Lemma~\ref{lemma_D}]
The first property follows from the quasi-concavity of $\alpha$, the second from the definition of the sets, for the third consider
\begin{align*}
\bigcap\limits_{\epsilon > 0} \mathcal{D}^\epsilon &= \{ D \in \mathcal{D} \; :  \; \alpha(D) \geq \alpha^* - \epsilon, \;\; \forall \epsilon > 0 \} = \mathcal{D}^*.
\end{align*}
\end{proof}

\begin{proof}[Proof of Lemma~\ref{lemma_ax}]
\begin{enumerate}
\item The positive value of the risk minimization problem $p(x) > 0$ means that all portfolios have positive risk at level $x$ -  that is $ \rho^x(D) > 0$ for all $D \in \mathcal{D}$. Therefore all portfolios have acceptability at most $x$ - that is $\alpha(D) \leq x$ for all $D \in \mathcal{D}.$ Consequently, $\alpha^* \leq x.$ We do not obtain a strict inequality as we have no information about the continuity of the risk measure in the parameter $x.$ 

\item The assumption of attainment of the infimum implies that there exists $\tilde{D} \in \mathcal{D}$ such that $\rho^x (\tilde{D}) \leq 0.$ Then $\alpha(\tilde{D}) \geq x$ and $\alpha^* \geq x.$

\item The maximal acceptability $\alpha^*$ above $x$ means that there exists some portfolio $\tilde{D} \in \mathcal{D}$ with $\alpha(\tilde{D}) > x.$ From the monotonicity of the family of risk measures and~\eqref{AIRM} it follows that $\rho^y(\tilde{D})\leq 0$ for all $y \leq x$ and therefore $p(y) \leq 0.$

\item The maximal acceptability $\alpha^*$ below $x$ means that $ \alpha(D) < x$ for all $D \in \mathcal{D}.$ Consequently, for all $D \in \mathcal{D}$ it holds $\rho^x(D) > 0.$ Since by Assumption~\ref{assumption} the infimum of the risk minimization problem is attained, also $p(x) > 0.$ The same follows for all $y \geq x$ as the family of risk measures is increasing.
\end{enumerate}
\end{proof}

\begin{proof}[Proof of Lemma~\ref{lemma_alg}]
\begin{enumerate}
\item Let $\alpha^* < \underline{\alpha}$. With the halving update rule, at iteration $n$ (counting from $0$) the tested value is $x_{n} =2^{-n} \cdot x_0$ and by Lemma~\ref{lemma_ax} at each step $p(x_n) > 0.$ Therefore after $\bar{M}$ iterations the algorithm is terminated and no non-zero lower bound on the acceptability is found. The portfolio $\bar{D}$ is never assigned, as no portfolio with a known lower bound on the acceptability is found.

\item Let $\alpha^* > \overline{\alpha}$. With the doubling update rule, at iteration $n$ (counting from $0$) the tested value is $x_{n} =2^{n} \cdot x_0$ and by Lemma~\ref{lemma_ax} at each step $p(x_n) \leq 0.$ Therefore, after $\bar{M}$ iterations the algorithm is terminated and no finite upper bound on the acceptability is found. The optimal solution $D^{x_{\bar{M}-1}}$ of the risk minimization problem $P(x_{\bar{M}-1})$ is outputted as $\bar{D}$. It has a degree of acceptability of at least $\overline{\alpha} = x_{\bar{M}-1}$.

\item (a)  According to Lemma~\ref{lemma_ax}  for $\alpha^* \in (\underline{\alpha}, \overline{\alpha})$ it holds $p(\underline{\alpha}) \leq 0 < p(\overline{\alpha}),$ therefore Step 1 of Algorithm~\ref{alg} identifies both lower and upper bound, $x_L$ and  $x_U.$ Lemma~\ref{lemma_ax} also guarantees that the found values are true bounds, $x_L \leq \alpha^* \leq x_U.$ Step 2 continues until the length of the interval is sufficiently small. 

(b) The optimal solution to the risk minimization problem $P(x_L)$ is returned as $\bar{D}$. By Assumption~\ref{assumption} it holds $\rho^{x_L} (\bar{D}) \leq 0,$ so $\alpha(\bar{D}) \geq x_L.$ From part (a) it follows that $x_L > \alpha^* -\epsilon,$ so $\bar{D} \in \mathcal{D}^\epsilon.$ 

(c) The worst-case scenario for the length of the interval after Step 1 is $x_L = x_0 \cdot 2^{\bar{M}-2}, x_U = x_0 \cdot 2^{\bar{M}-1}.$ Since the bisection step decreases the length of the interval by half, after $i$ bisection iterations the length of the interval would be $x_0 \cdot 2^{\bar{M}-2-i}$. To obtain a length below $\epsilon$ we need $i > \log_2 \frac{x_0}{\epsilon} + \bar{M} - 2$.
\end{enumerate}
\end{proof}

\begin{proof}[Proof of Lemma~\ref{lemma_convergence}]
By Lemma~\ref{lemma_alg} the algorithm for the tolerance $\epsilon$ outputs $D^\epsilon \in \mathcal{D}^\epsilon.$ Since $\mathcal{D}$ is compact, there is a subsequence  $\{ D^{\epsilon_{nk}} \}_{nk \in \N}$ with a limit, denoted $\tilde{D}$, in the feasible set. The compactness also implies there exists $C < \infty$ such that $\vert D \vert \leq C$ for all feasible positions $D \in \mathcal{D}$. Then, since $\vert \frac{1}{C} D^{\epsilon_{nk}} \vert \leq 1$ and $\frac{1}{C} D^{\epsilon_{nk}} \xrightarrow{p} \frac{1}{C} \tilde{D},$ scale invariance and the Fatou property of $\alpha$ imply for any fixed $\delta > 0$
\begin{align*}
\forall \epsilon_{nk} \leq \delta: \;\; \alpha (D^{\epsilon_{nk}}) \geq \alpha^* - \delta   \;\; \Rightarrow \;\; 
\alpha (\tilde{D}) \geq \alpha^* - \delta.
\end{align*}
Letting $\delta$ go to zero, we obtain $\alpha(\tilde{D}) \geq \alpha^*$. Therefore, $\tilde{D}$ is an element of the set $\mathcal{D}^*.$
\end{proof}

\section{Proofs from Section~\ref{sec_dynamic}}
\label{AppendixB}
\begin{proof}[Proof of Lemma~\ref{lemma_W}]
For the first part, positive homogeneity follows from the self-financing property and the linearity of the portfolio value $V_s.$ As for the second part, the recursiveness, we have 
\begin{align*}
\mathcal{H}_t (V_t) := \{ (h_s)_{s = t, \dots, T-1} \; \vert \;& \trans{\mathbf{1}} h_s = V_s, V_{s+1} = \trans{R_{s+1}} h_s, s = t, \dots, T-1 \} \\
= \{ (h_s)_{s = t, \dots, T-1} \; \vert \;& \trans{\mathbf{1}} h_t = V_t, V_{t+1} = \trans{R_{t+1}} h_t, \\
& \trans{\mathbf{1}} h_s = V_s, V_{s+1} = \trans{R_{s+1}} h_s, s = t+1, \dots, T-1 \} \\
= \{ (h_s)_{s = t, \dots, T-1} \; \vert \;& h_t \in H_t(V_t),  (h_s)_{s = t+1, \dots, T-1} \in \mathcal{H}_{t+1} (\trans{R_{t+1}} h_t) \}.\\
\end{align*}
The result for the set $\mathcal{H}_t^+ (V_t)$ is obtained similarly. 
\end{proof}

\begin{proof}[Proof of Theorem~\ref{thm_alpha}]
First we note that $\alpha^*$ is scale invariant, i.e. $\alpha^*_t(\lambda V_t) = \alpha^*_t(V_t)$ for any $\lambda >0, \lambda\in\cF_t, t\in\cT$, which follows immediately from the scale invariance of the DCAI $\alpha$ and Lemma~\ref{lemma_W}(1). Thus, it is enough to prove that $\alpha^*_0(1) = \alpha_t^*(1; \omega)$, for all $t\in\cT$ and $\omega\in\Omega$, which we will show next.

To prove the second claim we, use two types of sets of risks: the set of risks of one-step-ahead dividends,
\begin{align*}
Q^x_t (\omega) := \{ \rho^x_t (D^{[t+1]} (h_t); \omega) \; \vert \; h_t \in H_t^+ (1)  \},
\end{align*}
and the set of risks of feasible portfolios,
\begin{align*}
P^x_t (\omega) := \{ \rho^x_t (D^{[t+1, T]} (h); \omega) \; \vert \; h \in \mathcal{H}_t^+ (1)  \}.
\end{align*}
The fact that the asset returns are iid and the assumption that the one-step risk measures are identical imply that at a given level $x$ the sets of one-step-ahead risks coincide across all times and all states,
\begin{align}
\label{eqB5}
Q^x_t (\omega_1) = Q^x_s (\omega_2) \text{ for all } s, t \in \cT \text{ and all } \omega_1, \omega_2 \in \Omega.
\end{align}
At time $T-1$ the two types of sets of risks for a given level $x$ coincide, $Q^x_{T-1} (\omega) = P^x_{T-1} (\omega)$.
The relationship between the acceptability index $\alpha$ and the corresponding family of risk measures $(\rho^x)_{x \in (0, \infty)}$ implies the following two equivalence:
\begin{align}
\begin{split}
\label{eqB6}
\alpha^*_t (1; \omega) &\leq \beta \; \Leftrightarrow \; \forall y > \beta: \; P^y_t (\omega) \cap \mathbb{R}_- = \emptyset, \text{ and} \\
\alpha^*_t (1; \omega) &\geq \beta \; \Leftrightarrow \; \forall x < \beta: \; P^x_t (\omega)  \cap \mathbb{R}_- \neq \emptyset.
\end{split}
\end{align}

We prove the claim by a backward induction.  Let $\bar{\omega} \in \Omega$ be an arbitrary state of the world and set $\alpha^* := \alpha^*_{T-1} (1; \bar{\omega})$. In the first step of the induction we prove that $\alpha^*_{T-1} (1; \omega) = \alpha^*$ for all states $\omega \in \Omega$: Consider a level $y > \alpha^*$. According to~\eqref{eqB6} the set $P^y_{T-1} (\bar{\omega}) = Q^y_{T-1} (\bar{\omega})$ contains positive elements only. Then,~\eqref{eqB5} implies that the same is true for the set $P^y_{T-1} (\omega) = Q^y_{T-1} (\omega)$, which means $\alpha^*_{T-1} (1; \omega) \leq \alpha^*$. Now consider a level $x < \alpha^*$. According to~\eqref{eqB6} the set $P^x_{T-1} (\bar{\omega}) = Q^x_{T-1} (\bar{\omega})$ contains some non-positive element. By~\eqref{eqB5}, the same is true for the set $P^x_{T-1} (\omega) = Q^x_{T-1} (\omega)$, so $\alpha^*_{T-1} (1; \omega) \geq \alpha^*$.

The induction hypothesis assumes that $\alpha^*_{s} (1) \equiv \alpha^*$ for all $s > t$. For levels $y > \alpha^*$ this means that the sets $P^y_{t+1} (\omega)$ and the sets $Q^y_t (\omega)$ (via~\eqref{eqB5} and the first step of the induction) contain positive elements only. For levels $x < \alpha^*$ this means that the sets $P^x_{t+1} (\omega)$ and the sets $Q^x_t (\omega)$ (via~\eqref{eqB5} and the first step of the induction) contain some non-positive element. The adaptiveness (R1) and independence (R2) of the risk measure imply that there exists an element $\bar{p} \in P^x_{t+1}$ that is non-positive in all states of the world. The same is true also for the set $Q^x_t$.

Inductive step: The properties of the risk measure imply the following form of the set $P^x_t$,
\begin{align*}
P^x_t = \left\lbrace \rho^x_t \left( - \left( V_{t+1} (h_t) \cdot p - D_{t+1} (h_t) \right) 1_{\{t+1\}} \right) \; \vert \; h_t \in H_t^+ (1), p \in P^x_{t+1} \right\rbrace.
\end{align*}
Consider a level $y > \alpha^*$. According to the induction hypothesis all $p \in P^y_{t+1}$ are positive. Then, by applying the monotonicity (R3), an arbitrary element of $P^y_t$ can be bounded by
\begin{align*}
\rho^y_t \left( - \left( V_{t+1} (h_t) \cdot p - D_{t+1} (h_t) \right) 1_{\{t+1\}} \right) \geq \rho^y_t \left(  D_{t+1} (h_t) 1_{\{t+1\}} \right) = \rho^y_t \left(  D^{[t+1]} (h_t) \right).
\end{align*}
The risk $\rho^y_t \left(  D^{[t+1]} (w_t) \right)$ is an element of the set $Q^y_t$, so by the induction hypothesis it is positive in all states of the world. This shows $\alpha^*_{t} (1) \leq \alpha^*$.

Now consider a level $x < \alpha^*$. Consider the elements of $P^x_t$ of the form
\begin{align*}
\left\lbrace \rho^x_t \left( - \left( V_{t+1} (h_t) \cdot \bar{p} - D_{t+1} (h_t) \right) 1_{\{t+1\}} \right) \; \vert \; h_t \in H_t^+ (1) \right\rbrace,
\end{align*}
where $\bar{p} \leq 0$ is a non-positive element of $P^x_{t+1}$, whose existence is guaranteed by the induction hypothesis. Monotonicity of the risk measure bounds these risks by
\begin{align*}
\rho^x_t \left( - \left( V_{t+1} (h_t) \cdot \bar{p} - D_{t+1} (h_t) \right) 1_{\{t+1\}} \right) \leq \rho^x_t \left(  D_{t+1} (h_t) 1_{\{t+1\}} \right) = \rho^x_t \left(  D^{[t+1]} (h_t) \right).
\end{align*}
The risks $\rho^x_t \left(  D^{[t+1]} (w_t) \right)$ are elements of the set $Q^x_t$, and by the induction hypothesis at least one of them is non-positive. Therefore, the set $P^x_t$ contains at least one non-positive element and $\alpha^*_{t} (1) \geq \alpha^*$.
\end{proof}

\begin{proof}[Proof of Theorem~\ref{thm_alpha_solution}]
Firstly, note that the construction of the trading strategy $\bar{h}$ guarantees that it is adapted and feasible. We prove the claim via backward induction by showing that 
\begin{align*}
\rho^x_t (D^{[t+1, T]} (\bar{h})) \leq 0 \text{ for all } x < \alpha^*.
\end{align*}
This suffices to show that $\alpha_t (D^{[t+1, T]} (\bar{h})) \geq  \alpha^*$. Since $\alpha^*$ is a supremum, equality follows.

Consider time $T-1$. Optimality of the position $h^*$ and the positive homogeneity imply that 
\begin{align*}
\rho^x_{T-1} (D^{[T,T]} (\bar{h}); \bar{\omega}) \leq 0
\end{align*}
for all levels $x < \alpha^*$. The iid asset returns and the identical one-step risk measures together with the positive homogeneity give the same for all states $\omega \in \Omega$. 

The induction hypothesis assumes that the risk
\begin{align*}
\rho^x_{t+1} (D^{[t+2,T]} (\bar{h})) \leq 0 \text{ for all } x < \alpha^*.
\end{align*}
For the inductive step we use the recursiveness of the risk measure to express the time $t$ risk as
\begin{align*}
\rho^x_t (D^{[t+1,T]} (\bar{h})) = \rho^x_t \left( - \left( \rho^x_{t+1} (D^{[t+2, T]} (\bar{h}) ) - D_{t+1} (\bar{h}) \right) 1_{\{t+1\}} \right).
\end{align*}
At level $x < \alpha^*$ the induction hypothesis and the monotonicity provide a bound
\begin{align*}
\rho^x_t (D^{[t+1, T]} (\bar{h})) \leq \rho^x_t (D^{[t+1]} (\bar{h})) \leq 0.
\end{align*}
The inequality $\rho^x_t (\bar{D}^{t+1} (\bar{h})) \leq 0$ follows again from the iid asset returns, identical one-step risk measures, the positive homogeneity and the strategy $\bar{h}$ corresponding to the scaled position $h^*$. We conclude $\alpha_t (D^{[t+1, T]} (\bar{h})) \geq \alpha^*$.
\end{proof}

%
%
%
%
%
\begin{sidewaystable}
\center
\caption{Algorithm~\ref{alg} for $\ait$, $\glr$ and $\raroc$ in a toy market model. \label{tab1}}
\begin{tabular}{c | c c c c}
		& $\omega_1$		& $\omega_2$		& $\omega_3$		& $\omega_4$ \\ \hline
Asset 1	& $1.04 $			& $1.045$		& $0.98$		& $0.985$ \\
Asset 2	& $1.045$			& $0.975$			& $1.055$		& $0.98$ 
\end{tabular}
\subcaption*{Panel A: Return matrix $R$ in the toy market model (two assets and four states of the world).}
\bigskip
\begin{tabular}{c c c c c   c   c c c c c   c   c c c c c}
\cline{1-5} \cline{7-11} \cline{13-17}
\multicolumn{5}{c}{AIT}                               & $\;$ & \multicolumn{5}{c}{GLR}                               & $\;$ & \multicolumn{5}{c}{RAROC}                            \\
Iter  & $x_L$      & $x_U$      & $x$       & $p(x)$  &      & Iter  & $x_L$      & $x_U$      & $x$       & $p(x)$  &      & Iter  & $x_L$      & $x_U$      & $x$       & $p(x)$ \\ \cline{1-5} \cline{7-11} \cline{13-17} 
\multicolumn{5}{c}{Step 1}                            &      & \multicolumn{5}{c}{Step 1}                            &      & \multicolumn{5}{c}{Step 1}                           \\ \cline{1-5} \cline{7-11} \cline{13-17} 
$1$   & $0$        & $\infty$   & $2$       & +       &      & $1$   & $0$        & $\infty$   & $2$       & $-$     &      & $1$   & $0$        & $\infty$   & $2$       & +      \\
$2$   & $0$        & $2$        & $1$       & +       &      & $2$   & $2$        & $\infty$   & $4$       & +       &      & $2$   & $0$        & $2$        & $1$       & +      \\ \cline{7-11}
$3$   & $0$        & $1$        & $0.5$     & $-$     &      & \multicolumn{5}{c}{Step 2}                            &      & $3$   & $0$        & $1$        & $0.5$     & $-$    \\ \cline{1-5} \cline{7-11} \cline{13-17} 
\multicolumn{5}{c}{Step 2}                            &      & $1$   & $2$        & $4$        & $3$       & $-$     &      & \multicolumn{5}{c}{Step 2}                           \\ \cline{1-5} \cline{13-17} 
$1$   & $0.5$      & $1$        & $0.75$    & $-$     &      & $2$   & $3$        & $4$        & $3.5$     & +       &      & $1$   & $0.5$      & $1$        & $0.75$    & $-$    \\
$2$   & $0.75$     & $1$        & $0.875$   & +       &      & $3$   & $3$        & $3.5$      & $3.25$    & +       &      & $2$   & $0.75$     & $1$        & $0.875$   & +      \\
$3$   & $0.75$     & $0.875$    & $0.8125$  & +       &      & $4$   & $3$        & $3.25$     & $3.125$   & $-$     &      & $3$   & $0.75$     & $0.875$    & $0.8125$  & $-$    \\
$4$   & $0.75$     & $0.8125$   & $0.7813$  & +       &      & $5$   & $3.125$    & $3.25$     & $3.1875$  & +       &      & $4$   & $0.8125$   & $0.875$    & $0.8438$  & +      \\
$5$   & $0.75$     & $0.7813$   & $0.7656$  & +       &      & $6$   & $3.125$    & $3.1875$   & $3.1563$  & +       &      & $5$   & $0.8125$   & $0.8438$   & $0.8281$  & +      \\
$6$   & $0.75$     & $0.7656$   & $0.7578$  & $-$     &      & $7$   & $3.125$    & $3.1563$   & $3.1406$  & $-$     &      & $6$   & $0.8125$   & $0.8281$   & $0.8203$  & $-$    \\
$7$   & $0.7578$   & $0.7656$   & $0.7617$  & $-$     &      & $8$   & $3.1406$   & $3.1563$   & $3.1484$  & +       &      & $7$   & $0.8203$   & $0.8281$   & $0.8242$  & +      \\
$8$   & $0.7617$   & $0.7656$   & $0.7637$  & $-$     &      & $9$   & $3.1406$   & $3.1484$   & $3.1445$  & +       &      & $8$   & $0.8203$   & $0.8242$   & $0.8223$  & +      \\
$9$   & $0.7637$   & $0.7656$   & $0.7647$  & $-$     &      & $10$  & $3.1406$   & $3.1445$   & $3.1426$  & $-$     &      & $9$   & $0.8203$   & $0.8223$   & $0.8213$  & $-$    \\
$10$  & $0.7647$   & $0.7656$   & $0.7651$  & $-$     &      & $11$  & $3.1426$   & $3.1445$   & $3.1436$  & +       &      & $10$  & $0.8213$   & $0.8223$   & $0.8218$  & +      \\
$11$  & $0.7651$   & $0.7656$   & $0.7654$  & +       &      & $12$  & $3.1426$   & $3.1436$   & $3.1431$  & +       &      & $11$  & $0.8213$   & $0.8218$   & $0.8215$  & +      \\
$12$  & $0.7651$   & $0.7654$   & $0.7653$  & $-$     &      & $13$  & $3.1426$   & $3.1431$   & $3.1428$  & $-$     &      & $12$  & $0.8213$   & $0.8215$   & $0.8214$  & $-$    \\
$13$  & $0.7653$   & $0.7654$   & $0.7653$  & $-$     &      & $14$  & $3.1428$   & $3.1431$   & $3.1429$  & +       &      & $13$  & $0.8214$   & $0.8215$   & $0.8215$  & +      \\ \cline{1-5} \cline{13-17} 
      & $0.76532$  & $0.76538$  &           &         &      & $15$  & $3.1428$   & $3.1429$   & $3.1429$  & +       &      &       & $0.82141$  & $0.82147$  &           &        \\ \cline{1-5} \cline{7-11} \cline{13-17} 
\multicolumn{5}{c}{$h^\epsilon = (55.17\%, 44.83\%)$} &      &       & $3.14282$  & $3.14288$  &           &         &      & \multicolumn{5}{c}{$h^\epsilon = (93.75\%, 6.25\%)$} \\ \cline{7-11}
      &            &            &           &         &      & \multicolumn{5}{c}{$h^\epsilon = (73.33\%, 26.67\%)$} &      &       &            &            &           &           
\end{tabular}
\subcaption*{Panel B: Iterations of Algorithm~\ref{alg} with input parameters $x_0 = 2, \epsilon = 10^{-4}$ and $\bar{M} = 15$. The last two rows give, respectively, the bounds $x_L$ and $x_U$ on the maximal acceptability, and an $\varepsilon$-optimal portfolio.}
\end{sidewaystable}

\begin{sidewaystable}
 \center
\caption{Iterations of the modified, the mixed and the zero-level version of Algorithm~\ref{alg} for $\glr$ in the market model from Table~\ref{tab1}, Panel~A ($\alpha^* = 3.1428$) with the tolerance $\epsilon = 10^{-4}$.
In the modified version the bisection is performed on the parameter $q = \frac{1}{2+x} \in [0, 0.5]$ after verifying the signs of $p(0)$ and $p(\infty)$. The termination criterion is set on the parameter $x$ to guarantee an $\epsilon$-solution is obtained. With the termination criterion on the parameter $q$ the algorithm would finish after $13$ iteration of Step 2, however, the interval for maximal acceptability would have length 1.6e-03. The mixed version switches to a bisection on the parameter $x$ as soon as a finite upper bound $x_U$ is obtained. -the zero-level version computes after each iteration the level $y$ for which the portfolio solving the risk minimization problem has zero risk. This level is used as a lower bound. The algorithm is run with initial parameters $x_0 = 2, \epsilon = 10^{-4}$ and $\bar{M} = 15$.\label{tab3}}

\small
\begin{tabular}{cccccccccccccccccccc}
\cline{1-6} \cline{8-13} \cline{15-20}
\multicolumn{6}{c}{Modified algorithm for GLR}                   & $\;$ & \multicolumn{6}{c}{Mixed algorithm for GLR}               & $\;$ & \multicolumn{6}{c}{Zero-level algorithm for GLR}                                                                        \\
Iter   & $q_L$     & $q_U$     & $q$       & $x$       & $p(x)$  &      & Iter & $q_L$    & $q_U$    & $q$      & $x$      & $p(x)$ &      & Iter & $x_L$                & $x_U$                & $x$                  & $y$                  & $p(x)$               \\ \cline{1-6} \cline{8-13} \cline{15-20} 
\multicolumn{6}{c}{Step 1}                                       &      & \multicolumn{6}{c}{Step 1}                                &      & \multicolumn{6}{c}{Step 1}                                                                                              \\ \cline{1-6} \cline{8-13} \cline{15-20} 
       &           &           & $0$       & $\infty$  & +       &      &      &          &          & $0$      & $\infty$ & +      &      & $1$  & $0$                  & $\infty$             & $2$                  & $3.1429$             & $-$                  \\
       &           &           & $0.5$     & $0$       & $-$     &      &      &          &          & $0.5$    & $0$      & $-$    &      & $2$  & $3.1429$             & $\infty$             & $6.2857$             & $3.1429$             & +                    \\ \cline{1-6} \cline{8-13} \cline{15-20} 
\multicolumn{6}{c}{Step 2}                                       &      & \multicolumn{6}{c}{Step 2}                                &      & \multicolumn{6}{c}{Step 2}                                                                                              \\ \cline{1-6} \cline{8-13} \cline{15-20} 
$1$    & $0$       & $0.5$     & $0.25$    & $2$       & $-$     &      & $1$  & $0$      & $0.5$    & $0.25$   & $2$      & $-$    &      & $1$  & $3.1429$             & $6.2857$             & $4.7143$             & $3.1429$             & +                    \\
$2$    & $0$       & $0.25$    & $0.125$   & $6$       & +       &      & $2$  & $0$      & $0.25$   & $0.125$  & $6$      & +      &      & $2$  & $3.1429$             & $4.7143$             & $3.9286$             & $3.1429$             & +                    \\ \cline{8-13}
$3$    & $0.125$   & $0.25$    & $0.1875$  & $3.3333$  & +       &      & Iter & $x_L$    & $x_U$    & $x$      &          & $p(x)$ &      & $3$  & $3.1429$             & $3.9286$             & $3.5357$             & $3.1429$             & +                    \\ \cline{8-13}
$4$    & $0.1875$  & $0.25$    & $0.2188$  & $2.5714$  & $-$     &      & $3$  & $2$      & $6$      & $4$      &          & +      &      & $4$  & $3.1429$             & $3.5357$             & $3.3393$             & $3.1429$             & +                    \\
$5$    & $0.1875$  & $0.2188$  & $0.2031$  & $2.9231$  & $-$     &      & $4$  & $2$      & $4$      & $3$      &          & $-$    &      & $5$  & $3.1429$             & $3.3393$             & $3.2411$             & $3.1429$             & +                    \\
$6$    & $0.1875$  & $0.2031$  & $0.1953$  & $3.1200$  & $-$     &      & $5$  & $3$      & $4$      & $3.5$    &          & +      &      & $6$  & $3.1429$             & $3.2411$             & $3.1920$             & $3.1429$             & +                    \\
$7$    & $0.1875$  & $0.1953$  & $0.1914$  & $3.2245$  & +       &      & $6$  & $3$      & $3.5$    & $3.25$   &          & +      &      & $7$  & $3.1429$             & $3.1920$             & $3.1674$             & $3.1429$             & +                    \\
$8$    & $0.1914$  & $0.1953$  & $0.1934$  & $3.1717$  & +       &      & $7$  & $3$      & $3.25$   & $3.125$  &          & $-$    &      & $8$  & $3.1429$             & $3.1674$             & $3.1551$             & $3.1429$             & +                    \\
$9$    & $0.1934$  & $0.1953$  & $0.1943$  & $3.1457$  & +       &      & $8$  & $3.125$  & $3.25$   & $3.1875$ &          & +      &      & $9$  & $3.1429$             & $3.1551$             & $3.1490$             & $3.1429$             & +                    \\
$10$   & $0.1943$  & $0.1953$  & $0.1948$  & $3.1328$  & $-$     &      & $9$  & $3.125$  & $3.1875$ & $3.1563$ &          & +      &      & $10$ & $3.1429$             & $3.1490$             & $3.1459$             & $3.1429$             & +                    \\
$11$   & $0.1943$  & $0.1948$  & $0.1946$  & $3.1393$  & $-$     &      & $10$ & $3.125$  & $3.1563$ & $3.1406$ &          & $-$    &      & $11$ & $3.1429$             & $3.1459$             & $3.1444$             & $3.1429$             & +                    \\
$12$   & $0.1943$  & $0.1946$  & $0.1945$  & $3.1425$  & $-$     &      & $11$ & $3.1406$ & $3.1563$ & $3.1484$ &          & +      &      & $12$ & $3.1429$             & $3.1444$             & $3.1436$             & $3.1429$             & +                    \\
$13$   & $0.1943$  & $0.1945$  & $0.1944$  & $3.1441$  & +       &      & $12$ & $3.1406$ & $3.1484$ & $3.1445$ &          & +      &      & $13$ & $3.1429$             & $3.1436$             & $3.1432$             & $3.1429$             & +                    \\
$14$   & $0.1944$  & $0.1945$  & $0.1944$  & $3.1433$  & +       &      & $13$ & $3.1406$ & $3.1445$ & $3.1426$ &          & $-$    &      & $14$ & $3.1429$             & $3.1432$             & $3.1430$             & $3.1429$             & +                    \\
$15$   & $0.1944$  & $0.1945$  & $0.1944$  & $3.1429$  & +       &      & $14$ & $3.1426$ & $3.1445$ & $3.1436$ &          & +      &      & $15$ & $3.1429$             & $3.1430$             & $3.1430$             & $3.1429$             & +                    \\ \cline{15-20} 
$16$   & $0.1944$  & $0.1945$  & $0.1945$  & $3.1427$  & $-$     &      & $15$ & $3.1426$ & $3.1436$ & $3.1431$ &          & +      &      & \multicolumn{6}{c}{$(x_L, x_U) = (3.14286, 3.14295)$}                                                                   \\
$17$   & $0.1944$  & $0.1945$  & $0.1944$  & $3.1428$  & $-$     &      & $16$ & $3.1426$ & $3.1431$ & $3.1428$ &          & $-$    &      & \multicolumn{6}{c}{$h^\epsilon = (73.33\%, 26.67\%)$}                                                                   \\
$18$   & $0.1944$  & $0.1944$  & $0.1944$  & $3.1429$  & $-$     &      & $17$ & $3.1428$ & $3.1431$ & $3.1429$ &          & +      &      &      & \multicolumn{1}{l}{} & \multicolumn{1}{l}{} & \multicolumn{1}{l}{} & \multicolumn{1}{l}{} & \multicolumn{1}{l}{} \\ \cline{1-6}
\multicolumn{6}{c}{$(x_L, x_U) = (3.14285, 3.14290)$}            &      & $18$ & $3.1428$ & $3.1429$ & $3.1429$ &          & +      &      &      & \multicolumn{1}{l}{} & \multicolumn{1}{l}{} & \multicolumn{1}{l}{} & \multicolumn{1}{l}{} & \multicolumn{1}{l}{} \\ \cline{8-13}
\multicolumn{6}{c}{$q_U - q_L =$ 1.9e-06, $x_U - x_L =$ 5.0e-05} &      & \multicolumn{6}{c}{$(x_L, x_U) = (3.14282, 3.14288)$}     &      &      & \multicolumn{1}{l}{} & \multicolumn{1}{l}{} & \multicolumn{1}{l}{} & \multicolumn{1}{l}{} & \multicolumn{1}{l}{} \\
\multicolumn{6}{c}{$h^\epsilon = (73.33\%, 26.67\%)$}            &      & \multicolumn{6}{c}{$h^\epsilon = (73.33\%, 26.67\%)$}     &      &      & \multicolumn{1}{l}{} & \multicolumn{1}{l}{} & \multicolumn{1}{l}{} & \multicolumn{1}{l}{} & \multicolumn{1}{l}{}
\end{tabular}

\end{sidewaystable}

\begin{table}
	\center
	\caption{The behavior of Algorithm~\ref{alg} for various input parameters in a market model with $d=10$ assets with short-selling constraints.\label{tab2}}
	
	\bigskip
	\subcaption*{Panel A: $\ait$, maximal acceptability $\alpha^* = 25.45$.}
	\begin{tabular}{lccccccccc}
		\hline
		\multirow{2}{*}{$x_0$} & \multirow{2}{*}{$\epsilon$} & \multirow{2}{*}{$M$} & $\;\;$ & \multicolumn{2}{c}{Step 1} & $\;\;$ & \multicolumn{2}{c}{Step 2}    & Run time \\
		&                             &                      &        & Iter   & $[x_L, x_U]$      &        & Iter       & $x_U - x_L$      & (s)      \\ \hline
		$2$                    & $10^{-4}$                   & $15$                 &        & 5      & $[16, 32]$        &        & 18         & 6.1e-05          & 3.78     \\
		$20$                   & $10^{-4}$                   & $15$                 &        & 2      & $[20, 40]$        &        & 18         & 7.6e-05          & 3.40     \\
		$200$                  & $10^{-4}$                   & $15$                 &        & 4      & $[25, 50]$        &        & 18         & 9.5e-05          & 3.56     \\ \hline
		$2$                    & $10^{-8}$                   & $15$                 &        & 5      & $[16, 32]$        &        & 31         & 7.5e-09          & 6.22     \\ \hline
		$2^{20}$               & $10^{-4}$                   & $15$                 &        & 15     & $[0, 64]$        &        & \multicolumn{2}{c}{no Step 2} & 1.88     \\
		$2^{20}$               & $10^{-4}$                   & $30$                 &        & 17     & $[16, 32]$        &        & 18         & 6.1e-05          & 4.67 \\
		$2^{-10}$              & $10^{-4}$                   & $15$                 &        & 15     & $[16, \infty]$   &        & \multicolumn{2}{c}{no Step 2} & 4.61 \\
		$2^{-10}$              & $10^{-4}$                   & $30$                 &        & 16     & $[16, 32]$        &        & 18         & 6.1e-05          & 7.15     \\ \hline
	\end{tabular}

	\bigskip
	\subcaption*{Panel B: $\glr$, maximal acceptability $\alpha^* = 279.62$.}
	\begin{tabular}{l c c   c   c c   c   c c  c}
		\hline
		\multirow{2}{*}{$x_0$} & \multirow{2}{*}{$\epsilon$} & \multirow{2}{*}{$M$} & $\;\;$ & \multicolumn{2}{c}{Step 1} & $\;\;$ & \multicolumn{2}{c}{Step 2}    & Run time \\
		&                             &                      &        & Iter   & $[x_L, x_U]$      &        & Iter       & $x_U - x_L$      & (s)      \\ \hline
		$2$                    & $10^{-4}$                   & $15$                 &        & 9      & $[256, 512]$      &        & 22         & 6.1e-05          & 21.53    \\
		$20$                   & $10^{-4}$                   & $15$                 &        & 5      & $[160, 320]$      &        & 21         & 7.6e-05          & 17.27    \\
		$200$                  & $10^{-4}$                   & $15$                 &        & 2      & $[200, 400]$      &        & 21         & 9.5e-05          & 15.74    \\ \hline
		$2$                    & $10^{-8}$                   & $15$                 &        & 9      & $[256, 512]$      &        & 35         & 7.5e-09          & 30.40    \\ \hline
		$2^{25}$               & $10^{-4}$                   & $15$                 &        & 15     & $[0, 2048]$       &        & \multicolumn{2}{c}{no Step 2} & 6.50     \\
		$2^{25}$               & $10^{-4}$                   & $30$                 &        & 18     & $[0.5,1]$         &        & 22         & 6.1e-05          & 23.91    \\
		$2^{-10}$              & $10^{-4}$                   & $15$                 &        & 15     & $[16, \infty]$   &        & \multicolumn{2}{c}{no Step 2} & 13.41 \\
		$2^{-10}$              & $10^{-4}$                   & $30$                 &        & 20     & $[256, 512]$      &        & 22         & 6.1e-05          & 30.27    \\ \hline
	\end{tabular}

	\bigskip
	\subcaption*{Panel C: $\raroc$, maximal acceptability $\alpha^* = 279.62$.}
	\begin{tabular}{lccccccccc}
		\hline
		\multirow{2}{*}{$x_0$} & \multirow{2}{*}{$\epsilon$} & \multirow{2}{*}{$M$} & $\;\;$ & \multicolumn{2}{c}{Step 1} & $\;\;$ & \multicolumn{2}{c}{Step 2}    & Run time \\
		&                             &                      &        & Iter   & $[x_L, x_U]$      &        & Iter       & $x_U - x_L$      & (s)      \\ \hline
		$2$                    & $10^{-4}$                   & $15$                 &        & 2      & $[2, 4]$          &        & 15         & 6.1e-05          & 7.20     \\
		$20$                   & $10^{-4}$                   & $15$                 &        & 4      & $[2.5, 5]$        &        & 15         & 7.6e-05          & 9.41     \\
		$200$                  & $10^{-4}$                   & $15$                 &        & 8      & $[1.56, 3.13]$    &        & 14         & 9.4e-05          & 12.84    \\ \hline
		$2$                    & $10^{-8}$                   & $15$                 &        & 2      & $[2, 4]$          &        & 28         & 7.5e-09          & 11.07    \\ \hline
		$2^{20}$               & $10^{-4}$                   & $15$                 &        & 15     & $[0,64]$          &        & \multicolumn{2}{c}{no Step 2} & 10.55    \\
		$2^{20}$               & $10^{-4}$                   & $30$                 &        & 20     & $[2, 4]$          &        & 15         & 6.1e-05          & 19.59    \\
		$2^{-15}$              & $10^{-4}$                   & $15$                 &        & 15     & $[0.5, \infty]$   &        & \multicolumn{2}{c}{no Step 2} & 7.04     \\
		$2^{-15}$              & $10^{-4}$                   & $30$                 &        & 18     & $[2, 4]$          &        & 15         & 6.1e-05          & 13.41    \\ \hline
	\end{tabular}

\end{table}

\begin{table}
\center
\caption{A comparison of the different versions of the algorithm in a market with $d = 10$ assets and $\vert \Omega \vert = 1000$ states of the world both with and without short-selling. A tolerance $\epsilon = 10^{-4}$ is used for all algorithms, the original and zero-level version use $x_0 = 2$ and $\bar{M} = 15$. Obtaining the final approximation $[x_L, x_U]$ is denoted in the table by $\alpha^*$, values are listed to two decimal places. \label{tab4}}

\bigskip
\subcaption*{Panel A: $\ait$, the maximal acceptability with short-selling constraints ($h \geq 0$) is $\alpha^* = 25.45,$ without short-selling constraints ($h$ free) it is $\alpha^* =25.72$.}
\begin{tabular}{cccccccccc}
\hline
                            & \multirow{2}{*}{Algorithm} & \multicolumn{2}{c}{Step 1} & \multicolumn{2}{c}{Bisection on $q$} & \multicolumn{2}{c}{Bisection on $x$} & $x_U - x_L$ & Run time \\
                            &                            & Iter   & $[x_L, x_U]$      & Iter          & $[x_L, x_U]$         & Iter          & $[x_L, x_U]$         &             & (s)      \\ \hline
\multirow{4}{*}{$h \geq 0$} & Original                   & 5      & $[16, 32]$        &               &                      & 18            & $\alpha^*$           & 6.1e-05     & 3.32     \\
                            & Modified                   & 2      & $[0, \infty]$     & 23            & $\alpha^*$           &               &                      & 8.3e-05     & 3.67     \\
                            & Mixed                      & 2      & $[0, \infty]$     & 5             & $[15, 31]$           & 18            & $\alpha^*$           & 6.1e-05     & 3.90     \\
                            & Zero level                 & 3      & $[23.42, 46.84]$  &               &                      & 18            & $\alpha^*$           & 5.9e-05     & 2.96     \\ \hline
\multirow{4}{*}{$h$ free}   & Original                   & 5      & $[16, 32]$        &               &                      & 18            & $\alpha^*$           & 6.1e-05     & 4.89     \\
                            & Modified                   & 2      & $[0, \infty]$     & 23            & $\alpha^*$           &               &                      & 8.5e-05     & 5.11     \\
                            & Mixed                      & 2      & $[0, \infty]$     & 5             & $[15, 31]$           & 18            & $\alpha^*$           & 6.1e-05     & 5.09     \\
                            & Zero level                 & 3      & $[23.45, 46.89]$  &               &                      & 18            & $\alpha^*$           & 5.1e-05     & 4.36     \\ \hline
\end{tabular}

\bigskip
\subcaption*{Panel B: $\glr$, the maximal acceptability with short-selling constraints ($h \geq 0$) is $\alpha^* = 279.62,$ without short-selling constraints ($h$ free) it is $\alpha^* =288.88$.}
\begin{tabular}{cccccccccc}
\hline
                            & \multirow{2}{*}{Algorithm} & \multicolumn{2}{c}{Step 1} & \multicolumn{2}{c}{Bisection on $q$} & \multicolumn{2}{c}{Bisection on $x$} & $x_U - x_L$ & Run time \\
                            &                            & Iter  & $[x_L, x_U]$       & Iter          & $[x_L, x_U]$         & Iter          & $[x_L, x_U]$         &             & (s)      \\ \hline
\multirow{4}{*}{$h \geq 0$} & Original                   & 9     & $[256, 512]$       &               &                      & 22            & $\alpha^*$           & 6.1e-05     & 19.45    \\
                            & Modified                   & 2     & $[0, \infty]$      & 29            & $\alpha^*$           &               &                      & 7.4e-05     & 18.42    \\
                            & Mixed                      & 2     & $[0, \infty]$      & 8             & $[254, 510]$         & 22            & $\alpha^*$           & 6.1e-05     & 19.75    \\
                            & Zero level                 & 3     & $[279.62, 559.24]$ &               &                      & 22            & $\alpha^*$           & 6.7e-05     & 14.54    \\ \hline
\multirow{4}{*}{$h$ free}   & Original                   & 9     & $[256, 512]$       &               &                      & 22            & $\alpha^*$           & 6.1e-05     & 39.56    \\
                            & Modified                   & 2     & $[0, \infty]$      & 29            & $\alpha^*$           &               &                      & 7.9e-05     & 40.85    \\
                            & Mixed                      & 2     & $[0, \infty]$      & 8             & $[254, 510]$         & 22            & $\alpha^*$           & 6.1e-05     & 41.66    \\
                            & Zero level                 & 3     & $[288.88, 577.76]$ &               &                      & 22            & $\alpha^*$           & 6.8e-05     & 32.17    \\ \hline
\end{tabular}

\bigskip
\subcaption*{Panel C: $\raroc$, the maximal acceptability with short-selling constraints ($h \geq 0$) is $\alpha^* = 2.98,$ without short-selling constraints ($h$ free) it is $\alpha^* =3.08$.}
\begin{tabular}{cccccccccc}
\hline
                            & \multirow{2}{*}{Algorithm} & \multicolumn{2}{c}{Step 1} & \multicolumn{2}{c}{Bisection on $q$} & \multicolumn{2}{c}{Bisection on $x$} & $x_U - x_L$ & Run time \\
                            &                            & Iter    & $[x_L, x_U]$     & Iter          & $[x_L, x_U]$         & Iter          & $[x_L, x_U]$         &             & (s)      \\ \hline
\multirow{4}{*}{$h \geq 0$} & Original                   & 2       & $[2, 4]$         &               &                      & 15            & $\alpha^*$           & 6.1e-05     & 5.88     \\
                            & Modified                   & 2       & $[0, \infty]$    & 18            & $\alpha^*$           &               &                      & 6.0e-05     & 5.88     \\
                            & Mixed                      & 2       & $[0, \infty]$    & 2             & $[1, 3]$             & 15            & $\alpha^*$           & 6.1e-05     & 5.62     \\
                            & Zero level                 & 2       & $[2.98, 5.96]$   &               &                      & 15            & $\alpha^*$           & 9.1e-05     & 5.53     \\ \hline
\multirow{4}{*}{$h$ free}   & Original                   & 2       & $[2, 4]$         &               &                      & 15            & $\alpha^*$           & 6.1e-05     & 6.91     \\
                            & Modified                   & 2       & $[0, \infty]$    & 18            & $\alpha^*$           &               &                      & 6.3e-05     & 8.00     \\
                            & Mixed                      & 2       & $[0, \infty]$    & 3             & $[3, 7]$             & 18            & $\alpha^*$           & 6.1e-05     & 8.12     \\
                            & Zero level                 & 2       & $[3.08, 6.15]$   &               &                      & 15            & $\alpha^*$           & 9.4e-05     & 7.04     \\ \hline
\end{tabular}

\end{table}

\bibliographystyle{alpha}


\newcommand{\etalchar}[1]{$^{#1}$}

\end{document}